\documentclass{article}

\usepackage[ruled]{algorithm2e}

\usepackage{epsfig}
\usepackage{amssymb}
\usepackage{amsmath}

\newcommand{\N}{\mathbb{N}}

\newtheorem{definition}{Definition}
\newtheorem{proposition}{Proposition}
\newenvironment{proof}{{\bf Proof}:}{ QED. \\}
\newenvironment{acks}{{\bf Acknowledgements:} \begin{quote}}{\end{quote}}
\newtheorem{theorem}{Theorem}
\newtheorem{example}{Example}
\newtheorem{lemma}{Lemma}

\begin{document}

\markboth{B. Faltings et al.}{
Peer Truth Serum: Incentives for Crowdsourcing
Measurements and Opinions}

\title{Peer Truth Serum: Incentives for Crowdsourcing
Measurements and Opinions}
\author{BOI FALTINGS \\ Swiss Federal Institute of Technology (EPFL) \\
RADU JURCA \\ Google Z\"{u}rich \\
GORAN RADANOVIC \\ Harvard University}

\maketitle

\paragraph{Abstract:}

Modern decision making tools are based on statistical analysis of abundant data, which is often collected by querying multiple individuals. 
We consider data collection through crowdsourcing, where independent and self-interested agents, non-experts, report measurements, such as sensor readings, opinions, such as product reviews, or answers to human intelligence tasks. 
Since the accuracy of information is positively correlated with the effort invested in obtaining it, 
self-interested agents tend to report low-quality data. Therefore, there is a need for incentives that cover the cost of effort, while discouraging random reports. We propose a novel incentive mechanism called Peer Truth Serum that encourages truthful and accurate reporting, showing that it is the unique mechanism to satisfy a combination of desirable properties.

\section{Introduction}

``In God we trust, all others must bring data" \cite{ElementsStatisticalLearning} tweets Mike Bloomberg, former Mayor of New York City \cite{BloombergTweet}, to emphasize the commitment of his administration to define policies, and take actions, based on data. Mr. Bloomberg is not alone to embrace a data-driven governance model; companies, governments and private individuals all over the world rely increasingly often on hard evidence, measurements and surveys to guide important decisions.

Data collection, however, can be  expensive and slow. For this reason, a lot of recent research was devoted to {\em crowdsourcing} mechanisms, where  a crowd of independent self-interested participants contribute data (opinions, measurements, answers, etc) through an online system without being supervised by a central authority~(Figure~\ref{fig:crowdsourcing}). Examples are crowdsensing and community sensing~\cite{Crowdsensing,CommunitySensing}, product reputation systems \cite{Dellarocas} operated by sites such as Tripadvisor\footnote{www.tripadvisor.com} or Amazon\footnote{www.amazon.com}, human computation platforms such as Amazon Mechanical Turk, online opinion surveys, and prediction or decision markets~\cite{PredictionMarket,PredictionMarket2,PredictionMarketSurvey,DecisionMarket}.

We model the scenario of crowdsourcing measurements as a {\em center} that seeks to gain knowledge about a random variable $P$ that can only be observed through independent {\em agents}. For example, $P$ can be the weight of an object (referencing the first scientifically documented example \cite{Galton} reported by Galton), the reading of a sensor or a quality of a company's strategy. For practical reasons, observations of $P$ are assumed to take values in a discrete, finite set $X = \{x_1, x_2 , \ldots x_N \}$, such as positive integers for the weight of an object, one to five stars for a product rating, or yes/no answers to an intelligence task. Each agent may report any of these values. The center aggregates the reports and typically compensates the agents for their effort.

\begin{figure}

\centerline{\includegraphics[width=8cm]{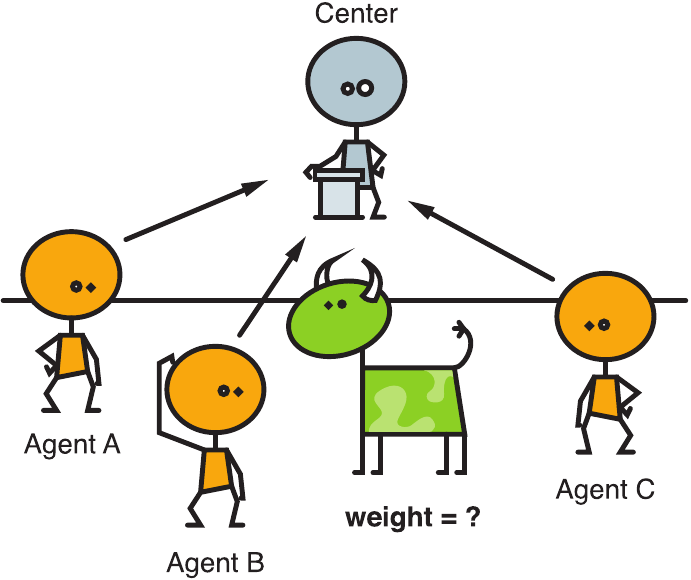}}

\caption{\it Crowdsourcing scenario. A center wants to determine the probability distribution of a random variable $P$, but can only observe it through independent {\em agents} reporting their estimates $r_A, r_B, ...$. In return, the center rewards the agents with payments $score_A,score_B,...$.}
\label{fig:crowdsourcing}
\end{figure}

In community sensing, agents install and maintain sensors and report their measurements to a center. The agents obtain different samples of the same phenomenon plus measurement noise, and the center aggregates them into a map. 
In human computation tasks, several crowdworkers are given the same task, such as counting objects in
an image, and their answers are aggregated into an average that eliminates individual inaccuracies. 
In public opinion polls, participants report their opinions or predictions about phenomena such as election outcomes, sports events, or engineering breakthroughs. The center publishes a continuously evolving aggregate of these reports. 

Studies and empirical observations suggest that crowdsourced agents often give inaccurate answers, for reasons of minimizing effort, incompetence or sloppiness \cite{difallah2012mechanical,ipeirotis2010quality}. Reward mechanisms should therefore be designed to maximize an agent's expected payment when reporting the value they privately believe to be the most accurate.

One class of {\em incentive-compatible} reward mechanisms is based on proper scoring rules \cite{savage,lambert}. These apply in applications where a ground truth eventually becomes public knowledge and can be used as reference to score the initial predictions: in Galton's example the bull was publicly weighed, and the participants were rewarded based on how close they came to the true weight. Proper scoring rules create strong incentives, in the sense that accuracy is a {\em dominant strategy}, optimal regardless of the reports submitted by other agents. Weather forecasting applications \cite{ScoringRuleWeather} and Prediction Markets \cite{PredictionMarket,PredictionMarket2,PredictionMarketSurvey,DecisionMarket} are prime examples of how scoring rules incentivize truth telling.

In many crowdsourcing scenarios a ground truth does not exist (e.g., the star-rating of a product) or will never be known. Therefore, 
another class of mechanisms takes the reports of one or several randomly
chosen {\em peers} as ground truth. If peers report accurately, this is also the best strategy for the agent being scored, and thus becomes an equilibrium. In this paper, we call such mechanisms {\em peer consistency} mechanisms. 

Peer consistency mechanisms have been used for various applications~(e.g. reputation reports~\cite{jf2003}). The most well-known characterization
is peer prediction~\cite{miller} which establishes a framework based on  proper scoring rules. 
 Peer consistency mechanisms have been proposed for sensor networks~\cite{Papakonstantinou-et-al}, for monitoring quality of service~\cite{JurcaBinderFaltings} and for incentivizing best and more accurate work in crowdsourcing~\cite{peer-consistency,DasguptaGoshwww13}. These mechanisms operate under different assumptions on the elicitation process, often assuming that all 
agents' beliefs are affected by measurements in an identical way (e.g. \cite{miller,BTS}), 
which is not realistic in practice. We investigate what is possible when this assumption is weakened or removed for a very general setting with 
non-binary answer space.

\section{Related Work}

The peer prediction mechanism \cite{miller} and 
the Bayesian truth serum (BTS) \cite{BTS} both have a setup of 
a standard Bayesian game where agents' prior information 
is defined by their type that they obtain in statistical
manner. Agents share a common belief about how they acquire 
their private information and thus their types. In peer prediction,
this common belief is also known to a mechanism designer, while
BTS removes the need of knowing it by eliciting 
additional information that indicates agents' posterior beliefs.   

Recently, many methods have been developed on the basis of
these two mechanism, making them more robust in different 
aspects, and often modifying the formal settings of these
methods. For example, robust versions of the (original) peer prediction 
were developed by \cite{jf2009}, with focus put on 
achieving robustness in terms of collusion resistance.

Robust Bayesian truth serums \cite{Witkowski-aaai12,W:14,rf2013,rf2014} extend the original 
version of the BTS mechanism to be applicable for small population of
agents \cite{Witkowski-aaai12} and elicitation of non-binary information \cite{rf2013,W:14,rf2014}.
The minimal truth serum by \cite{MinimalBTS} shows how to reduce the
amount of `redundant' information in a BTS elicitation process.

A peer prediction without a common prior by \cite{witkowski-ec12,W:14} uses a two stage 
elicitation process to remove the need of common prior information.
The mechanism is applicable to a setting with a clear temporal separation between periods prior to and posterior to 
the observation process. The knowledge free peer prediction \cite{knowledgeFreePP} also assumes a 
similar two stage setting, but unlike the mechanisms of \cite{witkowski-ec12}, it requires agents to 
have a common prior belief. Both mechanisms elicit agents beliefs  (posterior or prior) along with 
the desired information.   

There is also a line of work investigating settings where 
agents observe multiple a priori similar phenomena 
\cite{DasguptaGoshwww13,learningPrior,rf2015} - instead of observing a sample from one phenomenon, agents
observe samples of different phenomena. The setting removes the need
of agents having the common prior belief. Most of these mechanisms are, 
however, designed for elicitation of binary information \cite{DasguptaGoshwww13,learningPrior}.

Other related mechanisms include simple output agreement 
mechanisms (e.g. \cite{wcOutputAgreement}), that elicits common knowledge,
an opinion poll mechanism of \cite{truthfulSurveys}, 
that does not provide strict incentives without the presence 
of a trusted peer, and a mechanism for aggregating expert opinions of \cite{collectiveRevelation}, that assumes 
a common prior belief of a particular structure (e.g. Bernoulli distribution) that is known to the mechanism. 

We investigate a setting where agents, who arrive stochastically,  
observe one and the same phenomenon. We are interested 
in mechanisms that are {\em minimal}, meaning that they only elicit agents' private information, while not requiring knowledge of agents' beliefs. 
These conditions make the known mechanisms mentioned in this 
section inapplicable to our setting, which is in accordance with the impossibility 
results of \cite{rf2013}, and motivates the work in this paper. The main
idea is to build on the idea of \textit{helpful reporting} \cite{jurca-faltings-ec08,acmec-workshop-2011} to 
develop a unique mechanism that incentivizes agents to provide \textit{useful} information, under mild assumptions on agents' beliefs.

\section{Contributions}

With respect to the existing literature, this paper makes three
major novel contributions. 

The first is that we show a novel incentive mechanism for truthful information elicitation, the {\em Peer Truth Serum}. With respect
to peer prediction and the Bayesian Truth Serum, it replaces the
assumption of a common posterior belief, obtained after measurement,
with the assumption of a common prior belief and a restriction
on the belief update to ensure that it reflects the observation.

This is an important step forward as we know of no scenario where 
the common posterior assumption is reasonable, whereas we have found the common prior assumption to be applicable in many practical settings. We mention application to some of these settings in Section~\ref{sec:applications} of this paper: community sensing, human computation, and prediction polls.

The second contribution is to show that the Peer Truth Serum associated
with a particular restriction on the belief update is unique up to a
linear transformation. This allows us to prove several impossibility
theorems, which as a third and main result gives us a complete characterization
of the space of achievable properties of minimal elicitation mechanisms with
respect to different assumptions about belief updates.

\section{Setting}

We consider a {\em center} that is interested in obtaining information
about a phenomenon that can only be observed by a set of agents over which the center has no control. For example, an environmental phenomenon like pollution, rainfall or distribution of wildlife is observed by different people in different locations, or the customer satisfaction with a new product is observed by the customers themselves.

We model the phenomenon as a random variable $P$, and observation of agents as a noisy signal $o$ with $N$ possible values $X = \{x_1, \ldots, x_N\}$.%
Observations are independently generated according to a distribution $Q$ which is dependent on the state (value) of $P$. The distribution $Q$ can be any probability distribution function over the set $X$. 
We denote by $\Delta(X)$ a set of a fully mixed probability distribution functions\footnote{A probability function is fully mixed if its probability values are strictly positive.} over answer space $X$, and we assume that $Q \in \Delta(X)$.    

To incentivize the agents to accurately report their observations, the
center rewards their observations according to a peer consistency 
mechanism, defined below. The reward can be a monetary payment, or
another form or recognition that is valued by the agents.

\subsection{Peer consistency mechanisms}

We use the term {\em peer consistency} to denote a class of mechanisms
that works as follows. A {\em center} solicits information about random variable $P$ from {\em agents} in rounds, indexed by the variable $t = 1,2, \ldots$. In every round the center seeks answers from $M > 1$ different agents that observe $P$, and the communication between the center and any of the agents is private. 
When reporting observations, an agent may lie, but does not collude with other agents.  
At the end of the round, the center rewards each agent that contributed
a report using a payment function $\tau$ that takes into account one
or several peer reports that have been submitted about $P$ by other
agents. 
Importantly, agents participate in the mechanism only once; it is not
a repeated game.

The center aggregates the reports $r$ received until time period $t$ into a histogram $H^t = \{h_1..h_n|h_i = |\{r|(r=x_i) \wedge (time(r) \leq t)\}|\}$, where $time(r)$ is the time when $r$ was received. It should be
initialized to be a fully mixed histogram with small numbers, for example
all 1 or according to a prior distribution. The distribution $R^t$ is a normalized version of $H^t$.

In this paper, we consider in particular peer consistency mechanisms
with {\em continuous revelation}, where at the end of each round the center publishes $R^{t}$.%
We call $R^t$ the \textit{public distribution}, and
we often drop the index $t$ from $R^t$ when it refers to the current round. 
 
As compensation for information, the center offers payment to the agents,
using a payment function defined as follows:
\begin{definition}
A {\em payment function} is denoted by $\tau: X \times X \times \Delta(X) \rightarrow \Re$ and takes as arguments:
\begin{itemize}
\item the agent's own report $r$,
\item a report provided by another {\em peer} agent during the same round, called the {\em reference report} and denoted as $rr$,
\item a distribution $R$ over all possible values.
\end{itemize}
and computes a numeric reward $\tau(r,rr,R) \in \Re$ that the center
pays to the agent.
\end{definition}
We note that agents know the distribution $R^{t}$, but do not know
the reference report $rr$ or the agent that submitted it. The distribution $R^{t}$ makes the entire history of reports available 
to the payment function. 
We are interested in mechanisms that work even for small values of $M$,
possibly as low as $M=2$, to allow frequent updating, and 
therefore use only a single reference report.

\subsection{Reporting Strategies} 

A pure reporting strategy is a mapping $s$ that specifies what answer an agent will report given his observation, his prior belief and the distribution of reports submitted by other agents:
\begin{definition}
A {\em reporting strategy} is a function: $s : X \times \Delta(X) \times \Delta(X) \rightarrow X$ where $s(o, Pr, R) \in X$ is the report when the agent has observed $o$. A strategy is dependent on $R$ and
the agent's prior belief $Pr$ (although in general, the strategy depends on both the prior and the posterior belief, we only make explicit the dependency on the prior belief). 

For a {\em truthful} reporting strategy we have: $\forall o: s(o, Pr, R) = o$.

A {\em singleton} reporting strategy is a function that maps to a single value: $\forall o: s(o, Pr, R) = x $.
\end{definition}

The reporting strategy for each agent is their action in the {\em reporting game} played by the $M$ agents reporting in the current round. 

\subsection{Agent Beliefs}

We assume that agents participating in the mechanism are rational and risk-neutral. Before starting a task, all agents have the same 
information and thus a {\em common} prior belief
$Pr[\cdot] \in \Delta(X)$ regarding the distribution of agents' observations. Each agent then draws an individual sample of the phenomenon, which could be an observation, an active measurement, or the result of a more complex investigation, and concludes on the value of its sample $o \in X$. After the observation is made, the agent privately updates her belief about the distribution of agents' observations to a new and private distribution $Pr[\cdot|o] \in \Delta(X)$. Figure~\ref{fig:beliefs} shows how an agent forms its posterior belief.

\begin{figure}

\centerline{\epsfig{file=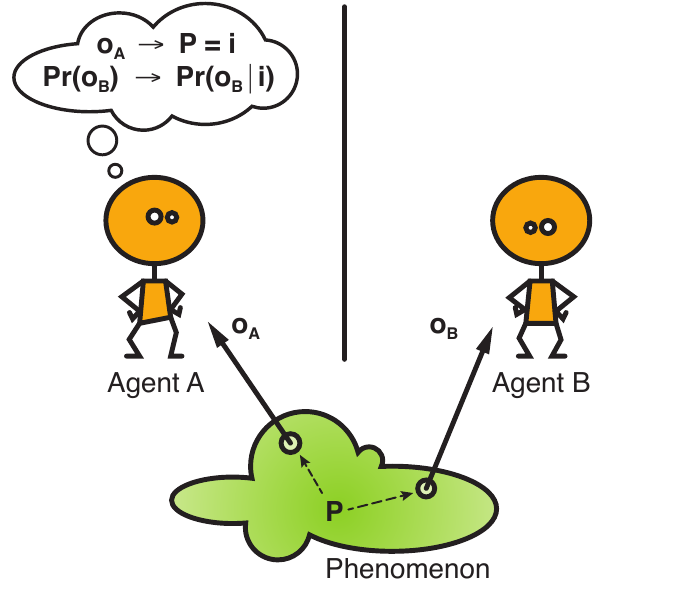,width=8cm}}

\caption{\it Before observation, agent $A$ will have a common {\em prior} belief about phenomenon $P$, in particular about distribution $Q$, and thus about the observation $o_B$ and the report $r_B$ it will be scored against. Upon observing the evidence $o_A$, agent $A$ forms a {\em posterior} belief about $P$, in particular about $Q$, and consequently also about $r_B$. For example, if $o_A$ is indicative of $P=i$, and agent $A$ believes that agent $B$ is a peer agent observing the same phenomenon, she will also increase its belief that $B$ will observe $o_B=o_A$ and report $r_B=o_A$.}
\label{fig:beliefs}
\end{figure}

We assume that agents who report in round $t$ have a \textit{common} prior belief $Pr^{t}[\cdot]$, which is typical for many peer prediction mechanisms. The center may not know this prior belief. In the course
of using the mechanism we describe, the public distribution $R^{t}$ 
aggregates an increasing number of peer reports and thus becomes an increasingly good approximation of the true distribution of
agent reports, as rational agents will adopt it as their prior belief.

When there are few or no prior reports, $R^{t}$ could be very 
different from $Q$ and also $Pr$. In this case, we assume that the
common prior will be {\em informed} with respect to the public
distribution. To characterize the difference, we first define:
\begin{definition}
A probability distribution $R$ is {\em $\rho$-close} to distribution $P$ if and only if:
\[
(\forall x) r(x) \in [(1-\rho) p(x) .. (1+\rho) p(x)] 
\]
\end{definition}
and define informedness as follows:
\begin{definition}
An agent's prior belief $Pr$ about an observation with the true distribution $Q$ is {\em informed} with respect to a distribution $R$ 
if and only if $\forall x \in X: %
(R[x] - Q[x])(R[x] - Pr[x]) \geq 0$.
It is {\em $\rho$-informed} if and only if it is either informed, or
$R$ is $\rho$-close to $Pr$.
\label{def:informed-prior}
\end{definition}
This is based on the observation that if agents do not agree with the public prior $R$, they should have
better information about the phenomenon and thus a more accurate belief.
For example, they may note that a fire generates a lot of smoke and thus
makes previous pollution readings inaccurate.

Following their observation $o$, agents may update their beliefs
in different ways, not known to the center. For example, agents may have different confidence in their observations, and thus give them different weights when combining them with the priors. They may also consider
different correlations between values. As there is no communication among
agents, they may adopt very different posterior beliefs about $P$ and
the peer reports. This differentiates our setting from standard peer prediction (\cite{miller}), where agents are required
to also update their beliefs in an identical manner so as to arrive at common posterior beliefs where the scoring rule can be applied. 

Thus, we consider that agents have two pieces of private information:
\begin{itemize}
\item their {\em observation}, for which there is a common prior
belief $Pr^t$;
\item their {\em update function}, which is a private type $\theta$ and
is a mapping from the prior probability distribution $Pr^t$ to the
posterior $Pr(\cdot | o)$. 
\end{itemize}

However, in order for an incentive mechanism to work the belief update
must reflect the observation $o$; in particular the observed value $o$ should have a higher probability in the posterior than in the prior distribution. We consider two forms of this constraint given in
Definitions~\ref{def:self-dominating} and~\ref{def:self-predicting} below.

If the agent is very convinced of her measurement, she may make it the value with the highest probability in $Pr[x|o]$:
\begin{definition}
An agent's belief update is {\em self-dominating} if and only if the observed value $o$ has the highest probability among all possible 
values $x$:
\begin{equation}
Pr[o|o] > Pr[x|o] \;\; \forall x \neq o
\label{eq:self-dominating-condition}
\end{equation}
\label{def:self-dominating}
\end{definition}
This condition will hold whenever agents believe that they observe an
identical sample from the distribution, for example they count the number of customers in the same restaurant at a specific time, or they measure
the temperature at the same place and time.

However, in many cases this is obviously too strong. For example, a customer who is dissatisfied with a product that has a very high reputation may believe that he has received a bad sample, and thus could still believe that the majority of customers are satisfied. Thus, we introduce a weaker condition by only requiring the {\em increase} in the probability of the observed value to be highest. 

\begin{definition}
An agent's belief update is {\em self-predicting}
if and only if the observed value has the highest relative increase in probability among all possible values: 
\begin{equation}
Pr[o|o]/Pr[o] > Pr[x|o]/Pr[x]  \;\; \forall x \neq o
\label{eq:self-predicting-condition}
\end{equation}
We call $\delta(o) = min_{y \neq o} \left ( \frac{Pr[o|o]}{Pr[o]}\frac{Pr[x]}{Pr[x|o]} -1 \right )$ the {\em gap} of the self-prediction
for observation $o$.
\label{def:self-predicting}
\end{definition} 

This condition allows for values to be correlated, so that beliefs about very similar answers may also increase. 
Bayesian updating of a categorical or a multinomial distribution with a Dirichlet prior (as in Dirichlet-categorical and Dirichlet-multinomial models)
is an example of a rational updating process that satisfies the self-predicting condition. The details of this updating model can be found in the proofs 
of Proposition \ref{prop:nopriorfree} and Proposition \ref{prop:no-arbitrage}. 

Note that for binary-valued signals, whenever the same observations of different agents are positively correlated, i.e. $Pr[o|o] > Pr[o]$,
the self-predicting condition is satisfied (see Section \ref{sec:binary}), while self-domination does not necessarily hold. In this case, the self-predicting condition is weaker than the self-dominating condition, but this does not have to be the case in general. Furthermore, notice that both of these conditions imply \textit{stochastic relevance} \cite{miller}.

It is possible to use alternative versions of the self-predicting
condition and obtain analogous mechanisms; we discuss such an alternative
later in the paper in Definition~\ref{def:self-predicting-quadratic}.

The main results of this paper are for the case where the public
distribution $R$ is a close approximation of the common prior $Pr$;
these are belief structures 1 and 2 below.

\begin{enumerate}

\item all belief updates are self-dominant, and this is known to all agents. Agents do not need to have a common prior or belief updating.
We will show that a simple mechanism is sufficient to ensure a truthful
equilibrium for this case.

\item all belief updates are self-predicting, and this is known to all
agents. The center knows the common prior within a small error tolerance.
This is the scenario for which we can show the strongest results.

\end{enumerate}

In the initial phases of the mechanism, the center may not have
a correct approximation of the agents' prior belief, and it is
interesting to analyze the behavior of mechanisms in such
belief structures:

\begin{enumerate}

\setcounter{enumi}{2}

\item all belief updates are self-predicting, and every agent believes
that its peers update their beliefs following the self-predicting
condition, and that some of them are truthful. The center does
not know the common prior. This is another, more general scenario for the initial stages of the mechanism, where we can show good behavior only
with additional restrictions on the value space.
\label{UntruthfulBeliefStructure}

\end{enumerate}

\paragraph{Common Knowledge}

Throughout the paper, we assume that agents have a common prior 
belief $Pr$ and that this is also common knowledge among all of them.
The only exception is the output agreement mechanism that does not require this assumption.

Furthermore, where we assume that agent priors are informed, or 
that agents have a particular belief structure, this
belief is also common knowledge to all agents. 

\subsection{Solution concepts}

Depending on the belief structure assumed, we can choose a suitable 
solution concept to the reporting game. The payoffs in the reporting
game depend on two uncertain elements:
\begin{itemize}
\item the observation of the peer agent, and
\item the resulting belief updates performed by the peer agent
that will determine its expected rewards.
\end{itemize}

Agents' beliefs about the peer observations will usually
depend on their observations. While the prior beliefs are assumed
common, even for the same observation the belief updates are
{\em subjective} and so the distributions $Pr(x|o)$ may be different for each agent. We cannot model this as a Bayesian game, but as a
subjective game.

We cannot model uncertainty about belief updates as a Bayesian game either, since there is no common prior distribution about these except
for the common knowledge that they satisfy the self-dominant or self-predicting condition, and it does not appear reasonable to assume
that agents would even form a subjective distribution. 
We model the belief updating function as an agent i's private type $\theta(i)$, and say that the type is {\em admissible}
if it satisfies the restrictions that are assumed as common knowledge, such as self-dominant or self-predicting. Thus, agents have no beliefs 
about other agents' belief updates except that they are of an admissible
type.

For the equilibrium, we use the concept of {\em ex-post subjective equilibrium} introduced by ~\cite{witkowski-ec12}:
\begin{definition}
Let $s_{i}$ be the reporting strategy of agent $i$ and $\hat{s}_{-i}$ be the profile of reporting strategies that agent $i$ believes of all other agents except $i$. In particular, let $\hat{s}_j(x,Pr,\theta_j,R) = s_j(x,Pr,R)$ be the reporting strategy adopted by agent $j$ with type $\theta_j$.
We say that the strategy profile $(s_{i}, \hat{s}_{-i})$ forms an {\em ex-post subjective equilibrium} if no agent $i$ can gain a higher expected payoff by deviating from $s_{i}$, for all admissible private types $\theta$ of any peer agent $j$, i.e.
\[
(\forall i,s',j,\theta_j) E_{Pr_i[x|o]} \tau(s_i(o,Pr,R), \hat{s}_j(x,Pr,\theta_j,R),R) \geq E_{Pr_i[x|o]} \tau(s'(o,Pr,R), \hat{s}_j(x,Pr,\theta,R),R)
\]
\end{definition}
We suggest this solution concept in particular for belief structures 1 and 2.

The uncertainty in belief structure~\ref{UntruthfulBeliefStructure} cannot in general be treated by an ex-post subjective equilibrium, as depending on their belief updates
different agents may believe that different underreported values are
their best response against truthful agents. Thus, in some cases no single non-truthful strategy can be an ex-post best response. However, because of their common prior, agents choose from the same set of underreported values that are eligible to be a best response to a
truthful peer.

The case where only one value is underreported occurs very frequently: when the phenomenon has a single predominant value, that value and no other will be underreported in the initial distribution. It also occurs when there are only 2 possible values. In this case, there is
only one possible non-truthful equilibrium.

\section{Properties of incentive mechanisms}

In this section we give an overview of desirable properties of an elicitation mechanism applicable to our setting. 

\subsection{Minimal Elicitation} 

\begin{definition}
We call an incentive mechanism for information elicitation {\em minimal} if only the elicited information is transmitted to the center.
\end{definition}
This property is important to maintain an efficient elicitation, and is different from existing mechanisms that also require agents to transmit information about their beliefs, whose size often significantly exceeds the elicited information itself. For example, in the Bayesian Truth Serum an agent has to submit a probability distribution over all possible answers in addition to indicating its own answer. This can be much more
information than the report itself and significantly inflate the effort
of the agent providing the information.

\subsection{Arbitrage-Free}

\begin{definition}
We call an incentive mechanism for information elicitation {\em arbitrage-free} if the expected payment of an agent who has no other information but the publicly available histogram $R$ is the same for all possible answers:
\begin{equation*}
\sum_{rr \in X} R[rr] \cdot \tau(r, rr, R) = C, \; \; \forall r \in X \\
\end{equation*}
where $C$ does not depend on reports $r$ and $rr$, but can depend on $R$.
\end{definition}
This property has important implications: agents who have only the publicly available prior information do not contribute to the accuracy of the mechanism, and thus should not have an incentive to participate. The arbitrage-free property allows to reduce their
expected payment to zero by subtracting $C$. 

\subsection{Unconstrained-prior}

\begin{definition}
An incentive mechanism for information elicitation is {\em unconstrained-prior} if
it places no restrictions on the prior beliefs of agents. 
\end{definition}
A mechanism that makes no assumptions about agents' prior beliefs is less susceptible to misreporting when such assumptions are not satisfied by all agents. We often consider this property together with 
informed priors according to Definition~\ref{def:informed-prior}. 

\subsection{Heterogeneous Update Tolerant}

\begin{definition}
An incentive mechanism for information elicitation is {\em heterogeneous update tolerant} if it does not require agents to update their beliefs in response to an observation in an identical manner.
\end{definition}
This property allows the population of agents to be diverse in the way they react to their measurements. 

However, within heterogeneous update tolerant mechanisms we will consider two types of restriction on the updates,
namely the self-dominating and self-predicting conditions given in
Definitions~\ref{def:self-dominating} through~\ref{def:self-predicting}.

\subsection{Truthful}

\begin{definition}
We call an incentive mechanism {\em truthful}, or {\em incentive-compatible}, if and only if truthful reporting is a strict equilibrium of the mechanism. 
\end{definition}

To exclude the trivial cases where the agent is indifferent between the equilibrium strategy $s_{i}$ and some other reporting strategy, we require the equilibrium to be {\em strict}. This also implies that strategies are pure because strict equilibria cannot contain mixed strategies. 

Note that in crowdsourcing settings where agents have to exert effort to find the correct answer, truthfulness also means maximizing their effort. 
As we require truthful reporting to be a strict equilibrium,
it is possible to scale the rewards so that the gain will offset the effort of truthful reporting \cite{miller,WB:13}.

\subsection{Asymptotically Accurate}\label{sec_aa}

The reason why we want truthful reports is to eventually obtain an accurate estimate of the true distribution $Q$ of observations. 
However, truthfulness is not a necessary conditions to obtain an
accurate estimate. As a weaker notion, we consider what is required so that over time, the aggregation of reports $R^{t}$ will converge to $Q$.
This also allows reports that are not truthful but {\em helpful} in the sense that they drive the aggregate towards the true distribution.

\begin{definition}
An incentive mechanism for information elicitation is said to be \textit{asymptotically accurate} if and only if it admits a strategy profile as a strict equilibrium
for which:
\[
\forall \epsilon > 0, \forall Q \in \Delta(X): \lim_{t \rightarrow \infty} Pr(\lVert R^{t}, Q \rVert \ge \epsilon) \rightarrow 0
\]
where $\lVert \cdot \rVert$ is $L_1$ norm between the vectors $R^{t}$ and $Q$.
\end{definition}

\begin{proposition}
Any truthful mechanism is asymptotically accurate.
\end{proposition} 
\begin{proof}
This proposition is trivially satisfied: if all agents report truthfully, the histogram of reports converges to the true distribution $Q$ due to the law of large numbers.
\end{proof}

The converse, however, is not true; there are asymptotically accurate mechanisms that do not have truthful equilibria. For example, a mechanism
where the center knows the true distribution $Q$ and incentivizes all
agents to report a value $x$ chosen such that $Q[x] - R^t[x] > 0$, would
have $R^t[x]$ converge to the true distribution, but not be truthful,
as $x$ is independent of the actual observation. However, such a 
mechanism has no practical use as the center has to already know the
distribution $Q$ to being with. 

Mechanisms with pointwise discontinuities can also be asymptotically accurate but have no truthful equilibrium. Consider
a mechanism $\tau'$ defined using a given asymptotically accurate mechanism $\tau$ as:
\begin{align*}
 \tau'(r, rr, R) = 
\begin{cases} 
	1 &\mbox{if } R[r] = \frac{1}{2} \textit{ and } r = rr \\ 
	0 &\mbox{if } R[r] = \frac{1}{2} \textit{ and } r \ne rr \\ 
	\tau(r, rr, R) &\mbox{otherwise}
\end{cases}
\end{align*}
The mechanism is asymptotically accurate, but is not truthful when $R[r]=\frac{1}{2}$. 
To rule out such cases, we impose a condition that 
agents' strategies are consistent whenever parameter $R$ 
and prior $Pr$ are sufficiently close. More precisely, if a reporting strategy $s$ is chosen
for some $R$ and $Pr$, then there must be a small interval in $R$ and $Pr$ that contains this
point and where the agent chooses the same strategy:
\begin{definition}   
An agent's strategy $s(o, Pr, R)$ is \textit{consistent} if there exist $\epsilon >0$ such that:
\begin{align}
\forall R', Pr' \in \Delta_{\epsilon}: s (o, Pr', R') = s(o, Pr, Pr)
\end{align}
where $\Delta_{\epsilon} = \{ F | F \in \Delta(X) \land  \lVert F, Pr \lVert < \epsilon \}$ and $\lVert . \lVert$ is $L_1$ norm.
\end{definition} 
Note that this definition allows the strategy to be different at the ends of the ranges $\delta$.

With these restrictions, we can now show that all
asymptotically accurate mechanisms must have a truthful equilibrium 
when the distribution $R$ is close to the agents' common prior 
belief $Pr$. 

\begin{proposition}\label{prop:truthrequired}
Suppose agents use pure consistent strategies. 
Then any asymptotically accurate mechanism with a continuous payment function that does not depend on the true distribution $Q$ must have an equilibrium strategy $s$ that is truthful whenever the private belief $Pr$ is sufficiently close to the public histogram $R$, i.e.:
\[
\exists \epsilon > 0 \; \; s.t. \lVert Pr, R \lVert < \epsilon \Rightarrow  \forall o \in X: s(o, Pr, R) = o \;
\]
where $\lVert . \lVert$ is $L_1$ norm. 
\end{proposition}     

{\bf Proof}: see Appendix.

Agents may misreport their observations if the public prior $R$ is sufficiently far from their prior beliefs. 
However, once the public histogram accurately approximates their prior beliefs, the agents report honestly. 

\section{Output Agreement}

When belief updates are heterogeneous and self-dominating, there is a simple class of mechanisms that are truthful independently of the prior distributions, and thus do not require the common prior assumption we
require for other settings in this paper. Following~\cite{vonAhn:2008:DGP:1378704.1378719}, we call these {\em output
agreement} mechanisms. Such mechanisms are widely proposed in the
literature, for example~\cite{jf2003,peer-consistency,Carvalho:2014:OMI:2692375.2692384}.

\begin{definition}
An {\em output agreement} mechanism is a peer consistency mechanism with the payment function 
\[
\tau(r, rr) = \begin{cases}
C \mbox{ if }  r=rr \\  0 \mbox{ otherwise}
\end{cases}
\]
where $C > 0$ does not depend on reports $r$ and $rr$.
\end{definition}

\begin{theorem}
There is a unconstrained-prior and truthful incentive mechanism for peer prediction that allows agents heterogeneous prior beliefs and heterogeneous updates satisfying the self-dominating condition.
\label{theo:self-dominating-exists}
\end{theorem}
\begin{proof}
In the peer consensus mechanism with $C = 1$, the expected reward for an agent with observation $x$ and report $y$ is equal to $Pr[y|x]$. By the self-dominating property~\ref{def:self-dominating}, $Pr[x|x] > Pr[y|x]$ for all $y \neq x$, and thus the expected payment is maximized by truthfully reporting $x$, independently of prior beliefs. Thus, the peer consensus mechanism is an example of a truthful mechanism that satisfies the theorem.
\end{proof}

\begin{example}
To illustrate the principle of the peer consensus scheme (with $C = 1$), consider a ternary answer space $X = \{x, y, z\}$, and suppose that an agent $i$ has a posterior and prior belief as shown in Figure \ref{fig:example-pc}. If agent $i$ believes that her peer $j$ is honest, her best response is also to be honest. For example, if agent $i$ observes $o_i = x$, her belief that her peer has also observed $x$ is equal to $Pr[o_j = x|o_i = x] = 0.7$. On the other hand, her belief that peer $j$ has observed $o_j = y$ becomes $Pr[o_j = y|o_i = x] = 0.3$. It is now easy to see that agent $i$'s expected score for reporting $x$ is $\mathbb E[\tau(x,o_j)|o_i = x] = 0.7$, and this is greater than her expected score for reporting $y$, $\mathbb E[\tau(y,o_j)|o_i = x] = 0.3$. 

\begin{figure}[h!]
\centering
\begin{tabular}{l|c|c|c|}
 & \multicolumn{3}{c}{Observation $o_j$} \\ 
 Belief & $x$ & $y$ & $z$\\ 
 \hline
 $Pr[o_j ]$ & 0.3 & 0.4 & 0.3 \\
 \hline
 $Pr[o_j | o_i = x]$ & 0.7 & 0.2 & 0.1  \\
 \hline
 $Pr[o_j |o_ i = y]$ & 0.1 & 0.8 & 0.1 \\
 \hline
 $Pr[o_j |o_ i = z]$ & 0.2 & 0.3 & 0.5
 \end{tabular}
\caption{self-dominating belief}
\label{fig:example-pc}
\end{figure}

\end{example}

Ideally, we would have an incentive mechanism that admits truthful reporting as an equilibrium without any conditions on the agents' prior beliefs. However, there could be two agents with prior beliefs
$Pr_1$ and $Pr_2$ such that $Pr_1$ updated according to an observation
$x$ is identical to $Pr_2$ updated according to an observation $y$.
A payment rule must give identical expected rewards to both agents
and thus incentivize them to report identical values, and so only
one of them can be truthful. More formally:
\begin{proposition}
\label{prop:nopriorfree}
There is no unconstrained-prior payment function $\tau$ such that the truthful reporting strategy is a strict equilibrium.
The statement remains true even if agents have common informed priors with belief updates satisfying the self-predicting condition.  
\end{proposition}

{\bf Proof}: see Appendix.

Although Proposition \ref{prop:nopriorfree} shows that with minimal mechanisms agents cannot be incentivized to always report honestly, 
Proposition \ref{prop:truthrequired} states that asymptotic accuracy requires truthfulness only when agents' prior and public prior $R$ coincide. This leads us to the concept of \textit{helpful} reporting strategies. 

\section{Helpful reporting strategies}\label{helpful-reporting-strategies}

The idea of a {\em helpful} strategy is that when the prior of an agent is far from the public prior, the agents may lie, but only in a way that brings the public histogram closer to their prior belief. From
the agent's perspective such reports help more than truthful reports
since they make the histogram converge faster to what it believes
to be the true distribution. Furthermore, 
if agents have informed priors, this type of lying strategy indeed helps the public histogram to get closer to the true distribution. 
This is conceptually similar to prediction markets \cite{PredictionMarket2}, where agents are rewarded for fixing inaccurate public information.
Note, however, that incentives in prediction markets are based on the ground truth, which never gets revealed in our setting.

We define a helpful reporting strategy by two constraints:
\begin{definition}
An agent with prior belief $Pr$ is said to follow a $\rho$-helpful reporting strategy with respect to public distribution $R$ if:
\begin{enumerate}
\item it adopts a truthful reporting strategy when the public distribution $R$ is $\rho$-close to the agent's private belief $Pr$, and
\item when the agent is not truthful, its strategy never reports an answer $x$ that is already over-represented in the current histogram:
\[
R[x] \ge Pr[x]  \Rightarrow \forall o \neq x : s(o, Pr, R) \neq x
\] 
\end{enumerate}
\label{def:helpful_strategy}
\end{definition}

Definition (\ref{def:helpful_strategy}) defines a family of strategies, not a single one. It explicitly leaves open the possibility for agents to non-truthfully report observations that are under-represented in the current public prior, and thus help drive this prior to the correct values more quickly. There are thus multiple equilibria and the mechanism
appears to be difficult to analyze. However, it turns out that as long
as agents' prior beliefs are informed with respect to the public
prior $R$, any equilibrium in helpful strategies makes $R$ converge
to the true distribution:

\begin{proposition}\label{prop:convergence}
For any true distribution $Q \in \Delta(X)$, any $\rho$-informed common prior belief $Pr[\cdot]$, the mechanism that supports $s$ as an equilibrium in $\rho$-helpful strategies is asymptotically accurate.
\end{proposition}

{\bf Proof}: see Appendix.

This remarkable properties lets us avoid analyzing the equilibria in
detail and focus on incentivizing helpful reporting strategies, which
then give us asymptotic accuracy.

\section{Peer Truth Serum}

We now develop an incentive mechanism called the {\em Peer Truth Serum} (PTS) that induces helpful reporting strategies under the self-predicting condition, which is a weaker notion than self-domination. We first consider necessary conditions that any incentive mechanism must satisfy to support an equilibrium in helpful strategies.

Firstly, a mechanism that admits helpful reporting needs to satisfy the \textit{arbitrage-free} property. The intuition is simple. When
the public prior $R[\cdot]$ reflects correctly an agent's prior $Pr[\cdot]$, the mechanism
has to be truthful by the first property of helpfulness. If an agent is allowed to have arbitrary posterior beliefs (satisfying the self-prediction),
the mechanism has to be truthful even when the agent's posterior 
approaches her prior $Pr[\cdot]$. This can happen, as formally shown by the following proposition, only when the expected payoff of an agent with belief equal to $R[\cdot]$ (e.g. the agent does not make any observation) is same for all possible reports. Note that we can always make $C=0$ by
subtracting the proper amount but this is not required to incentivize
helpful strategies.   
  
\begin{proposition}
\label{prop:no-arbitrage}
{\bf (arbitrage-free)} Let $\tau$ be a payment function accepting a helpful reporting strategy as a strict equilibrium. The expected payment of an agent who reports according to the public distribution $R$ is the same for all possible answers:
\begin{equation}
\sum_{rr \in X} R[rr] \tau(r, rr, R) = C, \; \; \forall r \in X \\
\end{equation}
\end{proposition}

{\bf Proof}: see Appendix.

Secondly, a mechanism that admits helpful reporting cannot be of an arbitrary structure - it can only reward an agent for agreement with her peer, with higher rewards assigned for a priori less likely reports. In particular, part of the reward function that depends on an agent's report is equal to $0$ if the reports of the agent  and her peer disagree, and otherwise, 
it is inversely proportional to the prior probability of the agent's report obtained from the public prior $R$. More formally: 

\begin{proposition} \label{prop:consensusonly}
{\bf (consensus only)}
Let $\tau$ be a bounded unconstrained-prior payment function accepting a helpful reporting strategy as a strict equilibrium.
Then $\tau$ can be written as $\tau(r,rr) = f(rr) + g(r,rr)$, where $g(r,rr) = 0$ when $r \neq rr$ and otherwise, when $r = rr$, $g(r, rr) = C/R[r]$. $C$ is strictly greater than $0$ and does not depend on agents' reports $r$ and $rr$.
\label{prop:consensus-only}
\end{proposition}

{\bf Proof}: see Appendix.

These two conditions motivate a new incentive mechanism that is minimal and tolerates heterogeneous belief updating within the limits of (\ref{eq:self-predicting-condition}). 
It is asymptotically accurate when agents have a common informed
prior, and truthful when this prior is close to the public distribution $R$. We call this scheme \textit{Peer Truth Serum}:

\begin{definition}
A {\em Peer Truth Serum} (PTS) 
is a peer consistency mechanism with the payment function 
\begin{align}%
\tau(r, rr , R) = f(rr) + \begin{cases}
\frac{C}{R[r]} &\mbox{ if $r=rr$}\\
0 &\mbox{otherwise}
\end{cases}
\end{align}
where $C > 0$ does not depend on reports $r$ and $rr$, and $f(rr)$ does not depend on report $r$. 
\end{definition}

The simplest form of the mechanism is obtained by setting $C = 1$ and $f(rr) = 0$, and we often use it in the following text. It should be noted that $C$ and $f(rr)$ are important for scaling purposes. For example, by putting $f(rr) = -C$, we can make expected payments equal to $0$ for agents whose only information is the public distribution $R$ (arbitrage-free condition). Furthermore, $C$ should not be dependent on agents' reports, but can be a function of a public prior $R$.
This is useful when one wants to bound payments, so that they take values in a specific interval. For example, we can set $C = \alpha \min_x R[x]$ and $f(rr) = \beta$, so that PTS score takes values in $[\beta, \beta + \alpha]$.

We can show that the Peer Truth Serum indeed supports equilibria
in helpful strategies:

\begin{proposition}\label{prop:ptshelpful}
The Peer Truth Serum 
admits a $\rho$-helpful reporting equilibrium, whenever agents have $\rho$-informed common prior belief and updates satisfying the self-predicting condition with $\rho < \min_{o} \frac{\delta(o)}{2 + \delta(o)}$.
\label{prop:payment-function}
\end{proposition}
\begin{proof}
We show that the payment function supports both conditions of a helpful reporting strategy for all belief updating functions that satisfy the self-predicting condition 
$\frac{Pr[o|o]}{Pr[o]} > \frac{Pr[y|o]}{Pr[y]}, o \neq y$, i.e. the posterior probability for $o$ increases the most when the agent actually observed $o$. 
Notice that function $f$ does not depend on what the agent who is rewarded reports, so it has no influence on the strategy and we
simply ignore it below.  

We first show that the mechanism has a truthful equilibrium when the public prior distribution $R$ is $\rho$-close to the private beliefs.
We observe that an agent who observes $o$ and reports $y$ expects
a reward:
\[
 \sum_{x \in X} Pr[x|o] \tau(y,x,R) =  Pr[y|o] \frac{C}{R[y]}
\]
provided that other agents are honest. 
In order for the mechanism to be strictly truthful, it has to hold for $y \neq o$:
\[
C \frac{Pr[y|o]}{R[y]} < C \frac{Pr[o|o]}{R[o]} \Leftrightarrow \frac{R[y]}{Pr[y|o]} > \frac{R[o]}{Pr[o|o]} 
\]
Using the fact that public $R$ is $\rho$-close to agents' prior $Pr$, 
we further obtain:
\[
\frac{Pr[y]}{Pr[y|o]}(1-\rho) > \frac{Pr[o]}{Pr[o|o]}(1+\rho) 
\]
Recall that the gap for value $o$ is defined as:
\[
\delta(o) = min_{y \neq o} \left ( \frac{Pr[y]}{Pr[y|o]}\frac{Pr[o]}{Pr[o|o]} -1 \right ) > 0
\]
then the truthfulness condition holds as long as $\rho$ satisfies:
\[
(\forall o) \frac{1 + \rho}{1 - \rho} < \delta(o) + 1
\]
or 
\[
\rho < \min_{o} \frac{\delta(o)}{2 + \delta(o)}
\]
Thus, when $R$ is $\rho$-close to $Pr$, the mechanism is truthful for all observations $o \in X$.
This equilibrium holds for any belief updating function $\theta$,
which means that it is an equilibrium.
This equilibrium in particular applies to belief structure 2 where
the center knows the agents' prior within a close bound $\rho$,
as well as to the case where $R$ has converged to be close enough
to the common prior.

When private beliefs are not close enough to the true distribution to
support a truthful equiilbrium, we need to consider the second condition.
Depending on the confidence of the agents, the threshold
$\rho$ for agents to be truthful varies. Thus, an agent can
expect that the peer agent is truthful with some probability $\alpha$.
Consider now the conditions for a strategy to be a best response
to a peer agent independently of the value of $\alpha$. 

Note that both the prior $Pr$ and the distribution $R$ are common and known to all agents. Thus, they all evaluate the second condition of
helpfulness in the same way, and partition the set of values into
the same set of underreported values where $R[x] \leq Pr[x]$ and overreported values where $R[x] > Pr[x]$. 

Consider first an agent who saves the effort of making an observation,
and reports according to its prior. If the peer agent is truthful,
the agent can always obtain a higher expected reward when reporting
an underreported value than an overreported value. If the peer agent
is not truthful, if it applies the same reasoning it also cannot be
expected to report an overreported value. Thus, the best response 
must be chosen within the space of helpful strategies.

Consider next an agent who has made an observation $o$ and thus reports
according to its posterior $Pr[\cdot|o]$ that it obtained by an update that satisfies the self-predicting condition. Suppose that reporting
value $x$ was not better than reporting $o$ according to the prior $Pr$:
\[
Pr[x]/R[x] \leq Pr[o]/R[o]
\]
then by the self-prediction condition we also have:
\[
Pr[x|o]/R[x] < Pr[o|o]/R[o]
\]
so that it cannot be better according to the posterior either.
Thus, the values that could be a best response under the posterior
distribution is a subset of the best responses under the prior distribution, and thus also within the space of helpful strategies.

It remains to be shown that there exist such equilibria. 
The following combination of strategies is a helpful reporting
equilibrium that can be easily implemented:
\begin{itemize}
\item if $R$ is sufficiently close to the prior $Pr$, observe
the phenomenon and report truthfully.
\item otherwise, pick the first of the underreported values according
to some order (numerical, lexicographic, etc.) and report this value,
indendently of the observation, or even without observing the phenomenon.
\end{itemize}
Such a singleton strategy is an equilibrium since any agent that deviates
from it cannot match a peer report, and thus obtains no reward at all.

\end{proof}

Untruthful helpful reporting strategies will often collapse into singleton
strategies, where the same value is reported independently of the
observation. This is because once an under-reported value $x$ is reported
in place of an observation $o$, the probability of the reference report
being equal to $x$ increases, thus making $x$ even more under-reported and causing other observations to also be misreported as $x$ in the best
response. For $k$ under-reported values, we thus have $k$ 
equilibria in helpful strategies where all agents adopt the singleton
strategy of reporting the same under-reported value. 

Note that a particular untruthful helpful equilibria is not 
necessarily the best response to a truthful peer, as the agent
may expect a higher reward from misreporting a different under-reported
value. Only when there is a single under-reported value can we
be sure that there is only a single non-truthful strategy that will
be followed by all agents. Fortunately this case is quite frequent:
besides the case where there are only 2 possible values, it is also
often the case that there is one predominantly likely value and this
will be the one that is underreported in an initial distribution.

In other cases, agents can only be certain that no peer will report an
over-reported value and thus report according to a non-helpful
strategy. Within this space, they will have to select among
multiple equilibria, as is the case with many other games as well.

In all cases, agents will use helpful strategies. No matter what
strategies they adopt, by Proposition~\ref{prop:convergence} the
mechanism will be asymptotically accurate.

Furthermore, we can show that PTS is the unique mechanism that admits helpful reporting and achieves asymptotic accuracy:
\begin{theorem}
The Peer Truth Serum is the unique arbitrage-free incentive mechanism for peer prediction that allows any heterogeneous updates satisfying the self-predicting condition and that is
\begin{itemize}
\item[a)] truthful when all agents hold the same prior distribution close to $R$ and 
\item[b)] asymptotically accurate when all agents have a prior distribution that is informed with respect to $R$.
\end{itemize}
\label{thm:pts}
\end{theorem}

\begin{proof}
Follows from Proposition~\ref{prop:convergence} and~\ref{prop:ptshelpful}. Uniqueness follows from Propositions~\ref{prop:consensus-only} and ~\ref{prop:payment-function}, as only consensus can be rewarded and
the payment function has to be the one used in PTS.
\end{proof}

\begin{example}
To illustrate the principle of the PTS, consider a ternary answer space $X = \{x, y, z\}$, the public distribution $R[x] = \frac{1}{3}$, $R[y] = \frac{1}{3}$ and $R[z] = \frac{1}{3}$, and the true distribution $Q[x] = 0.55$, $Q[y] = 0.4$ and  $Q[z] = 0.05$. Agent $i$'s best response to truthful behaviour of her peer $j$ depends on her private beliefs. Notice that we assume agents to have informed priors $Pr$, i.e. priors that are closer to the true distribution $Q$ than $R$ is. 

For example, suppose that an agent $i$ has a posterior and prior belief as shown in Figure \ref{fig:example-pts1}. Clearly, agent $i$'s prior is more informed than $R$ is. If agent $i$ observes $c$, her best response is to report $x$ or $y$, not $z$. Namely, her expected score for reporting $z$ is $E[\tau(z, o_j, R)|o_i = z] = \frac{0.2}{1/3}$, which is less than what she expects to get for reporting $x$ or $y$, e.g. $E[\tau(x, o_j, R)|o_i = c] = \frac{0.4}{1/3}$. However, reporting $x$ or $y$ brings $R$ closer to $Q$ (by updating $R$ with the report), which is not the case if agent $i$ is honest.  

\begin{figure}[h!]
\centering
\begin{tabular}{l|c|c|c|}
 & \multicolumn{3}{c}{Observation $o_j$} \\ 
 Belief & $x$ & $y$ & $z$ \\ 
 \hline
 $Pr[o_j ]$ & 0.5 & 0.4 & 0.1 \\
 \hline
 $Pr[o_j | o_i = x]$ & 0.7 & 0.2 & 0.1  \\
 \hline
 $Pr[o_j |o_ i = y]$ & 0.4 & 0.5 & 0.1 \\
 \hline
 $Pr[o_j |o_ i = z]$ & 0.4 & 0.4 & 0.2 \\
 \end{tabular}
\caption{Self-predicting belief (1st case)}
\label{fig:example-pts1}
\end{figure}

Now, suppose that an agent $i$ has a posterior and prior belief as shown in Figure \ref{fig:example-pts2}. Since the agent's private prior is close to the public distribution $R$, the agent's best response to the honest behaviour of her peer is to report truthfully. For example, when agent $i$ observes $z$, her expected score for reporting $z$ is equal to $E[\tau(z, o_j, R)|o_i = c] = \frac{0.5}{1/3}$, while for report $y$, it is equal to $E[\tau(y, o_j, R)|o_i = c] = \frac{0.3}{1/3}$. 

\begin{figure}[h!]
\centering
\begin{tabular}{l|c|c|c|}
 & \multicolumn{3}{c}{Observation $o_j$} \\ 
 Belief & $x$ & $y$ & $z$ \\ 
 \hline
 $Pr[o_j ]$ & 1/3 & 1/3 & 1/3 \\
 \hline
 $Pr[o_j | o_i = x]$ & 0.5 & 0.3 & 0.2  \\
 \hline
 $Pr[o_j |o_ i = y]$ & 0.3 & 0.5 & 0.2 \\
 \hline
 $Pr[o_j |o_ i = z]$ & 0.2 & 0.3 & 0.5 \\
 \end{tabular}
\caption{Self-predicting belief (2nd case)}
\label{fig:example-pts2}
\end{figure}

Since both helpful reporting and honest reporting make the public distribution $R$ converge towards the true distribution $Q$, it is intuitive that the mixture of the two strategies should also make $R$ converge towards $Q$. Therefore, although the PTS is not necessarily a truthful mechanism, it incentivizes agents to provide reports whose histogram approximates the true distribution $Q$.    

\end{example}

\section{Feasibility results}

The asymptotic accuracy of the PTS mechanism requires that we limit agents beliefs in two ways, by assuming that agents have informed priors and self-predicting updates. Although one might be tempted to relax these conditions, we prove that this is not possible. 
Namely, considering agents that use consistent pure strategies,
Theorem \ref{thm:pts} and Proposition \ref{prop:truthrequired} tell us that all asymptotically accurate mechanism have a form of the peer truth serum. 
We can, thus, use this fact to prove several impossibility results that explain why the conditions of Theorem \ref{thm:pts} cannot be further relaxed. 

We first state that some restriction on how agents update their priors is needed for asymptotic accuracy, e.g. the self-predicting condition. The claim is proven by using the result of Proposition \ref{prop:consensus-only}: an asymptotically accurate mechanism has to have a specific form, with the reward of an agent dependent on her report through the function $g(r, rr)$. Since truthfulness is required when the prior $Pr$ is equal to the public prior $R$, one can easily show that $\frac{Pr[x|x]}{Pr[x]} > \frac{Pr[y|x]}{Pr[y]}$ has to hold. In other words, agents cannot have arbitrary belief updating functions.

\begin{theorem}
There is no incentive scheme that allows unrestricted heterogeneous updates and is asymptotically accurate.
\label{theo:no-truthful}
\end{theorem}

{\bf Proof}: see Appendix.

We can show a similar result for the prior beliefs of the agents: in order to achieve asymptotic accuracy, the agents' prior belief needs to satisfy the informed prior condition. The intuition behind the result is the following. If agents are allowed to have an uninformed prior beliefs, the relationship between the public prior $R$, the private prior $Pr$ and the true distribution $Q$ can be one of the following: either $R$ is closer to $Q$ or $Pr$ is closer to $Q$. The agent, however, cannot distinguish between these two cases, so her report will innevitabily push $R$ further away from $Q$ in one of the two cases above.

\begin{theorem}
There is no unconstrained-prior asymptotically accurate incentive mechanism for self-predicting heterogeneous update functions.
\label{theo:no-general-prior}
\end{theorem}

{\bf Proof}: see Appendix.

Therefore, the main results of this paper include two theorems regarding the feasibility of minimal incentive mechanisms for peer prediction: Theorems~\ref{theo:no-truthful} and~\ref{theo:no-general-prior}. Together with Theorems~\ref{theo:self-dominating-exists} and~\ref{thm:pts}, these theorems define the space of feasible incentive mechanisms 
with respect to different degrees of unconstrained-priorness and tolerance to heterogeneous updates, as shown in Figure~\ref{fig:feasibility}.

\begin{figure*}
\centering

\begin{tabular}{|l|c|c|c|}
 & \multicolumn{3}{c}{Belief Updates} \\ 
 Prior & self-dominating & self-predicting & any \\ 
 \hline
 = R & truthful(Th~\ref{theo:self-dominating-exists}) & truthful(Th.~\ref{thm:pts}) & not asympt. acc. (Th.~\ref{theo:no-truthful}) \\
 \hline
 informed & truthful(Th~\ref{theo:self-dominating-exists}) & asympt. acc.(Th.~\ref{thm:pts}) & not asympt. acc. (Th.~\ref{theo:no-truthful}) \\
 \hline
 unconstrained-prior & truthful(Th~\ref{theo:self-dominating-exists}) & not asympt. acc. (Th.~\ref{theo:no-general-prior}) & not asympt. acc. (Th.~\ref{theo:no-truthful}) 
 \end{tabular}
\caption{Feasibility space for minimal mechanisms with heterogeneous belief updates.}
\label{fig:feasibility}
\end{figure*}

\section{Specific results for binary answer spaces}\label{sec:binary}

The requirement of Theorem \ref{thm:pts} is that agents at each round $t$ share a common prior belief, not necessarily known to the mechanism. The condition is tight, as shown by Proposition \ref{prop:common-prior}. The result is based on the fact that each agent might have a prior significantly different than the public prior $R$, but be convinced that the others have prior equal to $R$ and, hence, are truthful. Then we can construct a scenario with a non-binary answer space $\{x, y, z\}$, where: 
\begin{itemize}
\item $x$ is always truthfully reported when it is observed;
\item for $y$, public prior $R[y]$ oscillates around the true frequency $Q[y]$, but $x$ is reported instead of $y$ when $R[y] > Q[y]$;
\item $y$ is reported instead of $z$ when $R[y] < Q[y]$.
\end{itemize}
Although an over-reported value is never reported instead of an under-reported value, in the described scenario we achieve almost the same effect. This further implies that the frequency of reports $x$ is strictly greater than $Q[x]$, while the frequency of reports equal to $z$ is strictly less than $Q[z]$, so $R$ does not converge towards $Q$. 

\begin{proposition}\label{prop:common-prior}
Suppose agents are allowed to have different prior beliefs. Then, there is no asymptotically accurate incentive mechanism 
for self-predicting heterogeneous update functions, informed priors and the answer space $|X| \ge 3$.
\end{proposition}

{\bf Proof}: see Appendix.

Notice, however, that the proposition assumes non-binary answer spaces. As we show below, the common prior requirement is not needed when the answer space is binary. 
In fact, in a binary setting many conditions that are discussed in the paper are satisfied under reasonable updating assumptions. For example, whenever the same observations from different agents are positively correlated, the self-predicting condition holds.

\begin{definition}
We call observation $o$ {\em indicative} if and only if $Pr[o|o] > Pr[o]$. 
\end{definition} 

\begin{lemma}\label{lm_binary_sp}
Suppose $|X| = 2$. If observations are indicative, then the self-predicting condition holds.
\end{lemma}
\begin{proof}
Since $Pr[o|o] = Pr[o] + \epsilon$, with $\epsilon > 0$, and $Pr[o|o] + Pr[x|o] = 1$ for $x \ne o$, it follows that:
\begin{align*}
\frac{Pr[o|o]}{Pr[o]} = 1 + \frac{\epsilon}{Pr[o]}
> 1 > 1- \frac{\epsilon}{Pr[x]} = \frac{Pr[x] - \epsilon}{Pr[x]} = \frac{Pr[x|o]}{Pr[x]}
\end{align*}
\end{proof}

The self-prediction alone is not enough to prove asymptotic accuracy of PTS applied to a binary scenario. However, if agents' priors are informed, then the peer truth serum is asymptotically accurate, even when agents do not have common prior. 

\begin{proposition}
Consider a binary answer space $X = \{x, y\}$ and agents who have informed priors, not necessarily common, and indicative observations. 
Then the peer truth serum admits an equilibrium in which each agent either reports honestly or the under-reported value $\bar x \in X$ for which $Q[\bar x] \le R[\bar x]$.
Consequently, the peer truth serum is asymptotically accurate.
\end{proposition}
\begin{proof}
Without loss of generality we can assume that $Q[x] \ge R[x]$ and $Q[y] \le R[y]$. 
The first statement is now equivalent to the claim that $x$ will not be untruthfully reported instead of observation $o=y$.
Since the prior is not common, an agent has to reason about other agents' beliefs. 
However, due to informed priors, we know that the posterior belief of an agent regarding her peer 
reporting $x$ is less than or equal to $Pr[x|o]$ and her peer reporting $y$ is greater than or equal to $Pr[y|o]$.
Therefore, for observation $o = y$, we know that the expected payoff is at least 
$Pr[y|y]/R[y] \ge Pr[y|y]/Pr[y] > Pr[x|y]/Pr[x] \ge Pr[x|y]/R[x]$,
where the first and the third inequalities are due to the informed priors, and the second inequality is due to Lemma \ref{lm_binary_sp}. Notice that
the expected payoff for reporting $x$ is at most $Pr[x|y]/R[x]$.
In other words, the expected payoff for reporting $y$ is strictly greater than 
for reporting $x$, so agents who observe $y$ will report $y$, hence, helpful reporting is an equilibrium strategy.  

Asymptotic accuracy follows by the same argument for asymptotic accuracy of helpful reporting strategies in \ref{prop:convergence}. 
\end{proof}

\section{Optimality of the Peer Truth Serum}
\label{sec:optimality}

We have derived the Peer Truth Serum on the basis of desirable incentive
properties only. We now show that it fortuitously is also the unique
mechanism that optimizes the accuracy of the resulting predictions.

To measure the prediction quality of a public prior $R^t$ at time $t$, i.e. the utility of the center at time $t$, we use a scoring rule on a random sample $x_s$ from the true distribution $Q$.
In particular, we consider the logarithmic scoring rule (\cite{logscore}):
\begin{align*}
S(R, x_s) = \log(R[x_s])
\end{align*}
and the quadratic scoring rule (also called Brier score \cite{ScoringRuleWeather}):
\begin{align*}
S(R, x_s) = 2 \cdot R[x_s] - \sum_{x}R[x]^2
\end{align*}

To characterize the influence of a single report on $R$, we first have to 
describe how the center aggregates the reports it receives. Since $R$ is a normalized histogram of reports,
on receiving the $t$-th report $x$, the distribution $R$ is
updated as\footnote{For simplicity we ignored the influence of additive (Laplace) smoothing on $R$.}:
\begin{align*}
&R^{t+1}[x] = \frac{1}{t+1}(tR^{t}[x] + 1) = R[x] + \frac{1}{t+1}(1 - R[x]) \\
&R^{t+1}[y \ne x] = \frac{t}{t+1}R[y] = R[y] - \frac{1}{t+1}R[y]
\end{align*} 
In other words, $R$ is shifted in the direction of the reported value by $\epsilon$, where $\epsilon = \frac{1}{t+1}$.
The utility of the center after incorporating agent $i$'s report $x_i$ at time $t >> 1$ is approximately equal to:
\begin{align}\label{diff_shadowing}
S(R^{t+1}, x_s) \approx S(R^t, x_s) - \sum_{y \in X} \frac{\partial S(R^t, y)}{\partial y} \epsilon R^t[y] + \frac{\partial S(R^t, x_s)}{\partial x_i} \epsilon
\end{align} 
where we used the first order approximation (Taylor expansion). Indeed, $t >> 1$ implies that $\epsilon = \frac{1}{t+1}$ is small,
so this approximation is accurate. Hence, the gain of 
the center for incorporating the agent's report is approximately:
\begin{align}\label{ds_gain}
S(R^{t+1}, x_s) - S(R^t, x_s) \approx \frac{\partial S(R^t, x_s)}{\partial x_i} \epsilon - \sum_{y \in X} \frac{\partial S(R^t, y)}{\partial y} \epsilon R^t[y]
\end{align} 

We can incentivize agent $i$ to report the value that causes the highest gain by assigning her a score that is proportional to it. 
Notice that only the first part of the left hand side of \eqref{ds_gain} depends on agent $i$'s report $x_i$,
so the score can be defined as:
\begin{align*}%
\tau(x_{i},x_{s}, R^t) = \frac{\partial S(R^t, x_{s})}{\partial x_{i}}
\end{align*} 

Due to the fact that the center cannot directly sample the true distribution $Q$, it has to use the report of agent $i$'s peer instead. Provided that the peers are honest, their reports correspond to the samples of $Q$. Together with the definition of the logarithmic scoring rule, this gives us that the  mechanism which incentivizes agents to maximize the logarithmic gain of the center:

\begin{align}\label{pts}
\tau(r, rr, R) =
\begin{cases}
	\frac{1}{R[r]}  &\mbox{if } r = rr \\
	0 &\mbox{otherwise} 
\end{cases}
\end{align} 

Equation \eqref{pts} represents the peer truth serum from the previous section, 
with $f(rr) = 0$ and $C = 1$. 
A more general form can be derived using the appropriately scaled logarithmic scoring rule: 
\begin{align*}
S(R, r) = C \log(R[r])
\end{align*}  
and shifting the payments by $f(rr)$.
Similar normalization can be applied to the payment function derived
from the quadratic scoring rule.

Hence, we obtain the result:
\begin{theorem}
When the distribution $R$ is updated as the histogram of reports received 
by the center, and the quality of $R$ is evaluated using the logarithmic scoring rule and a peer report, the PTS mechanism is the unique mechanism that maximizes the expected quality of $R$ obtained by the center.
\label{theo:optimal-log-result}
\end{theorem}

The derivation above has some similarities with the shadowing method 
proposed by \cite{witkowski-ec12}, where an agent is rewarded using the shifted prior and the quadratic scoring rule.
Expression \eqref{diff_shadowing} can be considered as the
\textit{differential} shadowing over the logarithmic scoring rule.  
By identifying crucial parts of the shadowing score, we were able to 
significantly simplify the scoring function - we obtained the peer truth serum  and, hence, showed how it relates to the logarithmic gain of the center. 

Note that the same reasoning can be applied to the quadratic scoring
rule to obtain:
\begin{align}\label{quadratic-pts}
\tau(r, rr, R) =
\begin{cases}
	2 - 2R[r]  &\mbox{if } r = rr \\
	-2R[r] &\mbox{otherwise} 
\end{cases}
\end{align} 
It turns out that this leads to an alternative version of the Peer
Truth Serum. It would require a slightly different, linear self-predicting condition:
\begin{definition}
An agent's belief update is {\em self-predicting} with respect to the quadratic scoring rule if and only if the observed value has the highest relative increase in probability among all possible values: 
\begin{equation}
Pr[o|o]- Pr[o] > Pr[x|o] - Pr[x]  \;\; \forall x \neq o
\label{eq:linear-self-predicting-condition}
\end{equation}
\label{def:self-predicting-quadratic}
\end{definition} 
However, obviously Proposition~\ref{prop:consensusonly} will not hold
for this scheme, and so the uniqueness and impossibility results would
have to be proven in a different way. We do conjecture, however, that 
they hold, and in fact would hold if other scoring rules are used as a 
basis as well. Recent work~(\cite{FW:16}) has shown
new ways of understanding scoring rules that may form a better basis
for such development.

\section{Application Examples}
\label{sec:applications}

To understand the applicability of the PTS scheme to the practical scenarios it was intended for, we have developed several example applications that allow to evaluate the performance of the scheme.

\subsection{Community Sensing}

Maintaining a map of current levels of an environmental phenomenon such
as air pollution requires a network of sensors spread out in space.
It would make a lot of sense to have these sensors owned and operated
by individuals who get paid for delivering the measurement data, an
idea known as {\em community sensing}. The main difficulty with such
a scheme is that individuals may report inaccurate or even fictitious
data. PTS can solve this problem by rewarding accurate data. In 
simulations on a detailed model of air pollution based on actual
measurements~\cite{PTS-Sensors}, we have shown that using an adaptation
of the PTS scheme:
\begin{itemize}
\item reporting actual measurements achieves significantly higher rewards than reporting fictitious data such as constant values or the current predictions;
\item increasing accuracy results in increased rewards;
\item "gaming" the scheme requires collusion of a majority of agents.
\end{itemize}

\subsection{Human Computation}

It has become popular to use a crowd of relatively unskilled workers to
solve tasks that computers cannot solve, such as interpretation of images
or complex texts. By giving the same task to multiple workers and aggregating their answers, errors will cancel out. However, this 
"wisdom of the crowd" does not hold when answers are subject to systematic bias~\cite{PTS-hcomp}. For example, if 10 instances of a task had an answer
of 100, workers will tend to overestimate the answer if the 11th instance
has an answer of only 50. Here, a bonus scheme can be derived from PTS
that gives additional rewards to uncommon answers. Using experiments on
the Amazon Mechanical Turk platform, we showed that such a bonus scheme 
can indeed eliminate the bias in workers' answers~\cite{PTS-hcomp}.

\subsection{Online Polls}

As the cost of systematic polling increases, conducting polls through
the internet is increasingly attractive. However, as participants 
in these polls select themselves, they attract many participants with
ulterior motives, usually to tip the outcome of the poll in a particular
direction.

Prediction markets~\cite{PredictionMarket,PredictionMarket2,DecisionMarket} have been found to counter this problem by forcing participants to bet their own money on the outcomes they are predicting.
Only participants that actually contribute to the accuracy of the poll outcome can expect to make a profit, thus keeping out those that have nothing to contribute. However, a major problem with prediction markets is that they can only be applied when the true answer will eventually become known. 

The PTS scheme can be applied to this scenario by comparing the predictions against each other. This has the advantage that rewards can be paid before the actual outcome is known; in fact the true outcome never has to be known. We have implemented a 
platform called Swissnoise to test such a scheme. The platform offers
public access to questions such as:
\begin{itemize}
\item vote outcome: Will Scotland become independent?  
\item sports events: Who will win the 2014 soccer world cup?
\item breakthroughs: When will China succeed in landing on the moon?
\end{itemize}
Users can see a current prediction in the form of the probability of
each answer, and add their own prediction. We randomly distributed
users among both a classical prediction market and the peer truth serum,
and showed that both schemes achieve about the same accuracy of prediction~\cite{swissnoise}.

What is interesting, however, is that the PTS scheme can be applied
to questions where we never know the outcome, such as: "What would 
be the outcome of a vote for issue X" when no such vote is planned,
or "What is the quality of product X", and we can again expect the
same high quality of prediction.

\section{Conclusions}

Distilling knowledge out of data is one of the big mantras of our time. When data is collected through agents that may have their own self-interest, it is crucial to ensure that these agents make an effort to provide accurate and meaningful data, instead of random or even malicious information.

Peer consistency is a useful concept as it provides such incentives without the need for independent verification of the ground truth. We have analyzed the feasibility of incentive mechanisms for peer consistency and shown that there is a unique mechanism class (the Peer Truth Serum) which is also optimal. It is the first and only mechanism that satisfies some very intuitive and practical constraints:

\begin{itemize}

\item while it does require agents to have a common prior belief, it does not also require common posterior beliefs, as is the case in earlier peer prediction mechanisms such as~\cite{miller};

\item it does not provide any reward to agents without information, and thus does not require any restriction on participation;

\item when the agents have informed prior beliefs, the mechanism is asymptotically accurate and guaranteed to converge to the true distribution.

\end{itemize}

In experiments, we have shown the applicability of the reward mechanism to community sensing~\cite{PTS-Sensors}, to incentivizing human computation~\cite{PTS-hcomp} and to online opinion polls~\cite{swissnoise}, with very encouraging results. 
Another type of mechanism we have considered is a {\em crowdsourcing} version of PTS where $R^t$ is an average taken over the reports received 
from different but similar phenomena during the same round \cite{tistPaper} - this version is more robust to collusive behaviour.

\begin{acks}
We owe great thanks to the anonymous reviewers and editor whose insightful comments have greatly improved this paper.
\end{acks}

\bibliography{ptsarXiv17}
\bibliographystyle{plain}

\section{Proofs}

\subsection{Proof of Proposition \ref{prop:truthrequired}}

\begin{proof}
As we are proving a necessary condition,
without loss of generality we can assume that the population of agents is homogeneous, 
meaning that agents have a common prior belief $Pr$, a common belief updating 
function and they use symmetric strategies. 
Assume, by contradiction, 
that an asymptotically accurate mechanism does not have an equilibrium strategy 
that is truthful when the common private belief $Pr$ is sufficiently close to a certain 
value $\bar R$ of the public histogram, in particular, when $\bar R = Pr$. 
Due to the consistency of agents' strategies, this means that there exists a set
$\Delta_{\epsilon} = \{ F | F \in \Delta(X) \land  \lVert F, Pr \lVert < \epsilon \}$ 
such that for all $R^t \in \Delta_{\epsilon}$ 
agents are not incentivized to report truthfully, i.e. by a strict equilibrium, agents misreport.

Now, suppose that the true distribution $Q$ is equal to the common prior $Pr$, and the prior is fixed. 
By the law of large numbers, $R^t$ converges to $Pr$ 
only if the histogram of reports $R^t \in \Delta_{\epsilon}$ 
converges to $Pr$. If agents adopt symmetric and consistent pure strategies, this can happen in only two cases:
\begin{itemize}
\item The agents report honestly, however this contradicts the assumption that the agents misreport; or,
\item the agents misreport in a specific way to obtain the convergence, 
however this means that $R^t$ does not converge when we consider a different common prior 
$Pr' \in  \Delta_{\epsilon} \backslash \{Pr\}$ 
that is asymmetric ($\forall x \ne y: Pr'[x] \ne Pr'[y]$), and the true distribution
$Q = Pr'$. Namely, agents' strategies for $Pr'$ and $R^t$ are equal to those for 
$Pr$ and $R^t$, while
the law of large numbers tells us that
the frequency of reports $x$ converges towards the frequency of
observations $o$ for pure strategies $s$ 
only if $x = s(o, Pr', R^t \in  \Delta_{\epsilon})$ and $x \ne s(o', Pr', R^t \in  \Delta_{\epsilon})$, where $o' \ne o$. 
\end{itemize}
The second claim follows directly from the proof, since the proof meets the conditions of the claim. 
\end{proof}

\subsection{Proof of Proposition \ref{prop:nopriorfree}}

\begin{proof}
We construct two probability distributions $Pr_{1}[\cdot]$ and $Pr_{2}[\cdot]$ in $\Delta(X)$, the space of possible distributions over $X$, such that any payment that encourages truthful reporting under $Pr_{1}$ encourages lying under $Pr_{2}$.

We build the priors $Pr_1$ and $Pr_2$ as expectations derived from a belief model where all observation distributions are possible with a probability characterized by a Dirichlet distribution. For example, if $\omega \in \Delta(X)$ is a distribution over the set $X$, the probability of $\omega$ is characterized by the Dirichlet prior:
\begin{align*}
p(\omega) = \frac{1}{B(\alpha)} \prod_{i = 1}^{N} \omega_{i}^{\alpha_{i} - 1}
\end{align*}
where $\omega = (\omega_{1}, \omega_{2}, \ldots \omega_{N}) \in \Delta(X)$ is the vector of $\omega_i$ = probability of observation $x_i$, and $\alpha = (\alpha_{1}, \alpha_{2}, \ldots, \alpha_{N}) \in \N^{N}$, $\alpha_{i} > 1$ are the Dirichlet coefficients. $B(\alpha)$ is the multinomial beta function.

Given the model above, we can construct the prior $Pr$ as an expectation over all possible values $\omega$:
\[
Pr[x_i] = \int_{\omega \in \Delta(X)} p(\omega) \cdot \omega_i \cdot d\omega
\] 

Let us also assume that agents use pure Bayesian updating. Since the Dirichlet distribution is a conjugate prior for the multinomial process, the posterior is also a Dirichlet distribution with the following form:
\[
p(\omega | x_{k}) = \frac{1}{B(\alpha')} \prod_{i = 1}^{N} \omega_{i}^{\alpha'_{i} - 1}
\]
where the coefficients $\alpha'$ after the observation $x_k$ are related to the coefficients $\alpha$ through the following equation:
\[
\alpha'_{i} =
\left\{ 
	\begin{array}{ll}
		\alpha_{i}                                & \mbox{if} \; i \neq k \\
		\alpha_{i} + 1                         & \mbox{if} \; i = k
	\end{array}
\right.
\]

Let us choose two different Dirichlet distributions $p_{1}$ and $p_{2}$ (describing two different belief models) such that the coefficients $\alpha 1$ and $\alpha 2$ satisfy:
\[
\begin{array}{ll}
\alpha 2_{x} &= \alpha 1_{x} + 1 \\
\alpha 2_{y} &= \alpha 1_{y} - 1 \\
\alpha 2_{z} &= \alpha 1_{z} \; \forall z \in X \setminus \{x, y\}
\end{array}
\]
The prior $Pr_1$ generated by the model $p_1$ is clearly different from the prior $Pr_2$ generated by $p_2$. Nevertheless, $p_{1}(\omega | x) = p_{2}(\omega | y)$, and therefore the posterior beliefs will also be the same: $Pr_1[z | x] = Pr_2[z | y]$ for all $z \in X$.

Since the payment of an agent only depends on the report, the public histogram $R^t$, and the expected reference report, an agent observing $x$ with the prior belief $Pr_1$ and an agent observing $y$ with the prior $Pr_2$ will expect exactly the same reward. Consequently, honesty can only be incentivized in one of the two scenarios. 

Notice, even if payment $\tau$ was a function of agents' identities, the proof would still hold as long as the center does not know which agents adopt which priors (as assumed by our setting). Similarly, the second claim of the proposition is satisfied since the center does not know whether agents common prior is equal to $Pr_1$ or $Pr_2$, while both priors and the corresponding posteriors meet the conditions of the claim. 
\end{proof}
 
\subsection{Proof of Proposition \ref{prop:convergence}}

We first show that in a helpful reporting strategy, an agent with informed prior
will never non-truthfully report a value that will push the public prior $R$ further away from the true distribution
$P$, with distance being measured in the $L_1$ norm.

\begin{lemma}\label{lm_informed_priors}
Suppose agents have $\rho$-informed priors. In a $\rho$-helpful reporting strategy, there will never be a non-truthful report for an observation $x$ instead of another answer $y$ when $R[x] \ge Q[x]$.
\end{lemma}
\begin{proof}
If agents have $\rho$-informed but not informed priors, their reports
are truthful according to the first condition of helpfulness. Otherwise,
agents have informed priors, so that $R[x] \ge Q[x]$ implies $R[x] \ge Pr[x]$, and by the 2nd condition 
of helpfulness $x$ will not be non-truthfully reported. 
\end{proof}

\begin{proof}(Proposition~\ref{prop:convergence})
We need to show that $R^T$ converges (in probability) to $Q$, i.e. \\
$\lim_{T \rightarrow \infty} Pr( \lVert R^T - Q\lVert \ge \epsilon) \rightarrow 0$, for
any $\epsilon > 0$. 
Since $\sum_x R^T[x] = 1$ and $\sum_x Q[x] = 1$, it is sufficient to show that the probability of any particular value 
$x$ being overestimated ($R[x] > Q[x]$) goes to 0 as $t$ goes to infinity.
Due to the informed priors and Lemma \ref{lm_informed_priors}, 
we know that agents will not non-truthfully report an over-reported value. 
Therefore, in the worst case scenario, a certain value $x$ is 
always reported when it is underestimated ($R[x] < Q[x]$), regardless of agents' observations,
while when it is overestimated ($R[x] \ge Q[x]$), agents report it 
truthfully (i.e. $x$ is reported when $o = x$). 
As assumed by the setting, $R^t$ is updated after $M$ reports, where 
$M$ is fixed.  

To analyze this situation we investigate a random variable $Y^t$:

\begin{align*}
Y^t =
\begin{cases}
	R^0[x]  &\mbox{if } t = 0\\
	\frac{num(o = x)}{M} &\mbox{if } t > 0 \textit{ and } R^t[x] \ge Q[x] \\  
	1 &\mbox{if } t > 0 \textit{ and } R^t[x] < Q[x]  
\end{cases}
\end{align*}

where $num(o=x)$ counts the number of observations equal to $x$, among $M$ agents that report at time $t$. 
Let us also define $T'$ as maximal $t \in (0, T)$ such that $R^t[x] < Q[x]$, and 
if such $t$ does not exist we set $T' = 0$. 
For $T'' > T'$, $\sum_{t=0}^{T''} Y^t$ is upper bounded by $\sum_{t=0}^{T''} Z^t$,
where $Z^t$ is a random variable defined by:

\begin{align*}
Z^t =
\begin{cases}
	R^0[x]  &\mbox{if } t = T' = 0 \\
	Q[x] + \frac{1}{T'} &\mbox{if }  t \le T' \textit{ and } T' > 0 \\
	\frac{num(o = x)}{M} &\mbox{if } t > 0 \textit{ and } R^t[x] \ge Q[x] \\
\end{cases}
\end{align*}

The claim $\lim_{T \rightarrow \infty} Pr( \lVert R^T - Q\lVert \ge \epsilon) \rightarrow 0$ 
is now equivalent to $\lim_{T \rightarrow \infty} Pr( |\frac{\sum_{t=0}^{T-1} Z^t}{T} - Q[x]| \ge \epsilon) \rightarrow 0$ (this is simply because $\frac{\sum_{t=0}^{T-1} Z^t}{T}$ approaches $R^T$ as $T \rightarrow \infty$).
We have two cases:

\begin{itemize}
\item $T' = 0$. In this case, 
mean $\mu$ and variance $\sigma^2$ of $\frac{\sum_{t=0}^{T-1} Z^t(x)}{T}$ satisfy:
\begin{align*}
\mu &= \frac{R^0[x]}{T} + \frac{T-1}{T} Q[x]\\
\sigma^2 &= M \cdot (T-1)\frac{Q[x](1-Q[x])}{(M\cdot T )^2} \le \frac{Q[x](1-Q[x])}{T}
\end{align*}
The former is due to the linearity of expectations, while the latter is due to the fact that 
for $t > 0$ variables $Z^t$ are i.i.d., while $\frac{num(o = x)}{M}$ is an average of $M$ i.i.d. Bernoulli variables. 
Using Chebyshev's inequality, we obtain:
\begin{align*} 
Pr(|\frac{\sum_{t=0}^{T-1} Z^t}{T} - \frac{T-1}{T}  Q[x] - \frac{R^0[x]}{T}| \ge \epsilon) \le \frac{Q[x](1-Q[x])}{T\epsilon^2}
\end{align*} 
Which means that $\lim_{T \rightarrow \infty} Pr( |\frac{\sum_{t=0}^{T-1} Z^t}{T} - Q[x]| \ge \epsilon) \rightarrow 0$.
\item $T' > 0$. We distinguish $Z^t$ variables for which $t \le T'$ from those for which $t > T'$.
The average of the former have mean equal to $Q[x] + \frac{1}{T'}$ and variance equal to 0.
The latter have mean equal to $Q[x]$, while each of them have variance equal to $Q[x](1-Q[x])$. 
Hence, mean $\mu$ and variance $\sigma^2$ of $\frac{\sum_{t=0}^{T-1} Z^t}{T}$ satisfy:
\begin{align*}
\mu &= \frac{T'}{T}(Q[x] + \frac{1}{T'}) + \frac{T-T'}{T}Q[x] = Q[x] + \frac{1}{T}\\
\sigma^2 &= M \cdot T'\frac{Q[x](1-Q[x])}{(M \cdot T)^2} \le \frac{Q[x](1-Q[x])}{T}
\end{align*} 
Using Chebyshev's inequality, we obtain:
\begin{align*} 
Pr(|\frac{\sum_{t=0}^{T-1} Z^t}{T} - Q[x] - \frac{1}{T}| \ge \epsilon) \le \frac{Q[x](1-Q[x])}{T\epsilon^2}
\end{align*} 
Which means that $\lim_{T \rightarrow \infty} Pr( |\frac{\sum_{t=0}^{T-1} Z^t[x]}{T} - Q[x]| \ge \epsilon) \rightarrow 0$, and the proposition is proven.
\end{itemize}
\end{proof}

\subsection{Proof of Proposition \ref{prop:no-arbitrage}}

\begin{proof}
We will show that the proposition must hold for a particular class
of prior beliefs; since the payment function must hold for all such
beliefs this is sufficient to show this as a necessary condition.
To shorten the notation, we will drop the third argument $R$ from
the payment function and write $\tau(x,y)$ for $\tau(x,y,R)$.

As in Proposition \ref{prop:nopriorfree}, assume a prior belief derived from a model where every distribution is possible with the following Dirichlet probability:
\[
p(\omega \in \Delta(X)) = \frac{1}{B(\alpha)} \prod_{i = 1}^{N} \omega_{i}^{\alpha_{i} - 1}
\]
where $\alpha = (\alpha_{1}, \alpha_{2}, \ldots \alpha_{N}) \in \N^{N}$, $\alpha_{i} > 1$ are the Dirichlet coefficients. $B(\alpha)$ is the multinomial beta function, and let $\Sigma = \sum_{j = 1}^N \alpha_j$

According to this model, the prior assigns the following probability to each answer $x_i$ is:
\begin{equation}
Pr[x_i] =\frac{\alpha_i}{\sum_{j=1}^N \alpha_j} = \frac{\alpha_i}{\Sigma}
\label{eq:dirichlet}
\end{equation}
If the public distribution $R$ precisely mirrors this prior expectation, the helpful strategy is truthful, so any helpful payment mechanism must satisfy the truthful reporting constraints:

\begin{equation}
\sum_{z \in X} Pr[z|x] \tau(x, z) > \sum_{z \in X} Pr[z|x] \tau(y,z) \;\; \forall x \neq y \in X;
\label{eq:icx}
\end{equation}
\begin{equation}
\sum_{z \in X} Pr[z|y] \tau(y, z) > \sum_{z \in X} Pr[z|y] \tau(x,z) \;\; \forall x \neq y \in X;
\label{eq:icy}
\end{equation}

The posteriors beliefs derived from the Dirichlet prior have a very simple form:
\[
Pr[z|x] =
\left\{ 
	\begin{array}{ll}
		\frac{\alpha_z}{\Sigma + 1}                               & if \; z \neq x \\
		\frac{\alpha_z + 1}{\Sigma + 1}                        & if \; z = x
	\end{array}
\right.
\]

which can be re-written as a function of the prior:
\[
Pr[x|x] = Pr[x] + \frac{1 - Pr[x]}{\Sigma + 1}
\]
\[
Pr[z|x] = Pr[z] - \frac{Pr[z]}{\Sigma + 1}
\]

Making a simplifying notation $\tau_z = \tau(x,z) - \tau(y,z)$ the truthful reporting constraint (\ref{eq:icx}) becomes:
\[
\sum_{z \in X} Pr[z|x] \tau_z > 0
\]
\[
\left(Pr[x] + \frac{1 - Pr[x]}{\Sigma + 1}\right) \tau_x + \sum_{z \in X \backslash \{ x \} } \left(Pr[z] - \frac{Pr[z]}{\Sigma + 1} \right) \tau_z > 0
\]
\begin{equation}
\label{eq:icx2}
\frac{\Sigma}{\Sigma+1} \sum_{z \in X} Pr[z] \tau_z + \frac{1}{\Sigma + 1} \tau_x > 0
\end{equation}

Similarly, the constraint (\ref{eq:icy}) becomes:
\begin{equation}
\label{eq:icy2}
\frac{\Sigma}{\Sigma+1} \sum_{z \in X} Pr[z] \tau_z - \frac{1}{\Sigma + 1} \tau_x < 0
\end{equation}

Both (\ref{eq:icx2}) and (\ref{eq:icy2}) must be satisfied for all priors that create an expectation over the signals that is exactly equal to the public histogram $R$. So consider, for example, a family of Dirichlet priors $p^k(\omega)$ derived from (\ref{eq:dirichlet}) where the Dirichlet parameters are:
\[
\alpha^k_i = k \cdot \alpha_i
\] 
for the integers $k = 1,2,\ldots$. Each prior $p^k(\omega)$ generates the same expectation over the signals:
\[
Pr^k[x] = \frac{\alpha^k_x}{\sum_j \alpha^k_j} = \frac{\alpha_x \cdot k}{k \cdot \Sigma} = \frac{\alpha_x}{\Sigma}
\]
however, the posterior updates are increasingly smaller:
\[
Pr^k[x|x] = Pr[x] + \frac{1 - Pr[x]}{k \cdot \Sigma + 1}
\]
\[
Pr^k[z|x] = Pr[z] - \frac{Pr[z]}{k \cdot \Sigma + 1}
\]

Translating the conditions (\ref{eq:icx2}) and (\ref{eq:icy2}) for each prior in this family, we have that for every $k=1,2,\ldots$ the following inequalities must hold simultaneously:

\begin{equation}
\frac{k \cdot \Sigma}{k \cdot \Sigma+1} \sum_{z \in X} Pr[z] \tau_z + \frac{1}{k \cdot \Sigma + 1} \tau_x > 0
\end{equation}
\begin{equation}
\frac{k \cdot \Sigma}{k \cdot \Sigma+1} \sum_{z \in X} Pr[z] \tau_z - \frac{1}{k \cdot \Sigma + 1} \tau_x < 0
\end{equation}

Which can only happen when:
\[
\sum_{z \in X} Pr[z] \tau_z = 0
\]
and therefore:
\begin{align*}
\sum_{z \in X} Pr[z] \tau_z = \sum_{z \in X} Pr[z] \left( \tau(x,z) - \tau(y,z) \right) =\sum_{z \in X} R[z] \left( \tau(x,z) - \tau(y,z) \right) = 0
\end{align*}
which proves the proposition.
\end{proof}

\subsection{Proof of Proposition \ref{prop:consensusonly}}

Before proving the proposition let us prove the following statement.

\begin{lemma}
Let $\tau$ be a bounded unconstrained-prior payment function accepting a helpful reporting strategy as a strict equilibrium. Then for a setting with at least 3 values with prior probabilities bounded away from 0 and 1, the payment $\tau(r,rr,R)$ is independent of $r$ for all values of $r \neq rr$,
and thus can be written as $\tau(r,rr,R) = f(rr,R)$.
\label{lm:independent}
\end{lemma}
\begin{proof} 
We will simplify the notation and drop the argument $R$ from the
payment function, thus we write $\tau(r,rr)$ instead of $\tau(r,rr,R)$.
The payment function $\tau$ is bounded and positive by assumption, so that for all answers $r,rr \in X$, there is some positive number $\Theta$ such that:
\[
0 \leq \tau(r,rr) \leq \Theta
\]
Let $x$, $y$ and $z$ be three possible values and contrary to the claim, assume that for $z$ we have $\tau(x,z) - \tau(y,z) = \delta > 0$ (if $\delta < 0$ then switch $x$ and $y$). Furthermore, let $\gamma$ be the
largest real such that $Pr[\bar x] = R[\bar x] \in [\gamma, 1-\gamma]$, for all $\bar x \in \{x,y,z\}$. By the assumption that the priors are bounded away from 0 and 1,
$\gamma > 0$.

Now consider an agent that adopts the distribution $R$ as its prior ($Pr[\cdot] = R[\cdot]$) and updates its belief in the following way:
\begin{itemize}
\item $Pr[x|x] = Pr[x] (1+\gamma)$
\item $Pr[y|x] = Pr[y](1+\gamma) - \epsilon$
\item $Pr[z|x] = Pr[z] (1-\gamma \frac{Pr[x]+Pr[y]}{Pr[z]}) + \epsilon$
\end{itemize} 
where $\epsilon > 0$ is some positive value.

Since $\tau$ accepts a helpful reporting strategy as an equilibrium and the priors $Pr[\cdot]$ are exactly equal to the public distribution, the agent must find truthful reporting rational, i.e.:
\begin{align*}
&Pr[x|x] \tau(x,x) + Pr[y|x] \tau(x,y) + Pr[z|x] \tau(x,z) \\
&> Pr[x|x] \tau(y,x) + Pr[y|x] \tau(y,y) + Pr[z|x] \tau(y,z)
\end{align*}
while the no-arbitrage condition (Proposition~\ref{prop:no-arbitrage}) states:
\begin{align*}
&Pr[x] \tau(x,x) + Pr[y] \tau(x,y)  + Pr[z] \tau(x,z) \\
&= Pr[x] \tau(y,x) + Pr[y] \tau(y,y) + Pr[z] \tau(y,z)
\end{align*}

Subtracting $(1+\gamma)$ times the no-arbitrage-condition, the truthful reporting condition becomes:
\begin{eqnarray*}
& \epsilon (\tau(x,z) - \tau(x,y)) - \gamma(Pr[x] + Pr[y] + Pr[z]) \tau(x,z) \\
& >  \epsilon (\tau(y,z) - \tau(y,y)) - \gamma(Pr[x] + Pr[y] + Pr[z]) \tau(y,z)
\end{eqnarray*}
or equivalently, where we let $\delta = \tau(x,z) - \tau(y,z)$:
\[
\epsilon (\delta + \tau(y,y) - \tau(x,y)) > \gamma \delta (Pr[x]+Pr[y]+Pr[z])
\]
By choosing a sufficiently small $\epsilon$, since $\tau(y,y) - \tau(x,y)$ is a bounded value, and $(Pr[x]+Pr[y]+Pr[z]), \gamma$ and $\delta$ are all positive, we have a contradiction that can only be resolved when $\delta = 0$, and thus $\tau(x,z) = \tau(y,z)$. By symmetry, $\tau(x,y) = \tau(z,y)$
and $\tau(y,x) = \tau(z,x)$, and similarly when there are more than 3 values.
\end{proof}

\textbf{Proof of the proposition}: Now we are ready to prove the proposition.

\begin{proof}
Consider again a situation with 3 values $x$, $y$ and $z$.
Using Lemma~\ref{lm:independent} and the arbitrage-free condition (Proposition~\ref{prop:no-arbitrage}):
\begin{align*}
&R[x] \tau(x,x) + R[y] \tau(x,y)  + R[z] \tau(x,z) \\
&= R[x] \tau(y,x) + R[y] \tau(y,y) + R[z] \tau(y,z)
\end{align*}
simplifies to:
\begin{align*}
R[x] \tau(x,x) + R[y] \tau(x,y) = R[x] \tau(y,x) + R[y] \tau(y,y)
\end{align*}
so that
\begin{align}
&R[x] (\tau(x,x)-\tau(y,x)) = R[y] (\tau(y,y)-\tau(x,y)) \\
&= R[z](\tau(z,z)-\tau(y,z)) = C
\label{eq:rewritten-non-arbitrage}
\end{align}
where the last equality holds by symmetry. Using Lemma~\ref{lm:independent} and equation \ref{eq:rewritten-non-arbitrage}, we obtain:
\begin{align*}
&R[x](\tau(x,x)-\tau(y,x)) = R[x](g(x,x) + f(x) - f(x)) \\
&= R[x]g(x,x) = C
\end{align*}
and thus $g(x,x) = C/R[x]$. Notice, $C$ does not depend on agents' reports (but can depend on $R$). Furthermore, since $g(r,rr)$ is the only part of mechanism $\tau$ that is dependent on report $r$, it has to be that $C > 0$. Otherwise, for highly confident posteriors $Pr[\bar x|\bar x] = 1 - \epsilon$, where $0 < \epsilon << 1$ and $\bar x \in \{x, y, z\}$, agents would either: 
\begin{itemize}
\item be incentivized to misreport because $Pr[\bar x|\bar x] g(\bar x, \bar x) < Pr[\tilde x| \bar x] g(\tilde x, \tilde x)$ for $\tilde x \ne \bar x$, $C < 0$ and small enough $\epsilon$;
\item or would be indeferent between reporting different values for $C = 0$, so the equilibrium would not be strict.   
\end{itemize}
\end{proof}

\subsection{Proof of Theorem \ref{theo:no-truthful}}

\begin{proof}
It is enough to restrict our attention to mechanisms described by Proposition \ref{prop:consensus-only}, because all asymptotically accurate mechanisms are of such form. Furthermore, knowing that mechanisms have to be at least partially truthful by Proposition \ref{prop:truthrequired}, it is enough to show that for arbitrary heterogeneous updates one cannot construct truthful mechanisms
when prior $Pr$ is equal to $R$.

Assume that for some $x$, $y$, the self-predicting condition holds for
$x$, so that:
\[
Pr[x|x]/Pr[x] > Pr[y|x]/Pr[y]
\]
but it does not hold for $y$, so that:
\[
Pr[y|y]/Pr[y] < Pr[x|y]/Pr[x]
\]
Consider now an agent whose prior beliefs are exactly equal to $R$, and that endorses $x$. To compare the rewards it can expect from reporting $x$ or $y$, it suffices to consider the component $g$, as $f$ is independent of its report (see Proposition \ref{prop:consensusonly}). 
For an agent that endorses $x$ to report $x$, it must carry
a higher expected reward than $y$, i.e.:
\begin{eqnarray*}
Pr[x|x] g(x,x) & > & Pr[y|x] g(y,y) \\
Pr[x|x] \frac{C}{Pr[x]} & > & Pr[y|x] \frac{C}{Pr[y]}
\end{eqnarray*}
so that from the self-predicting assumption, it follows that $C>0$.

Conversely, for an agent that also believes $Pr = R$ and that observed $y$ to truthfully report $y$, the part corresponding to $g$ must be higher for reporting $y$ than for reporting $x$:
\begin{eqnarray*}
Pr[y|y] g(y,y) & > & Pr[x|y] g(x,x) \\
Pr[y|y] \frac{C}{Pr[y]} & > & Pr[x|y] \frac{C}{Pr[x]}
\end{eqnarray*}
so that with the self-predicting assumption violated, it follows that
 $C < 0$.

Thus, we have a contradiction.

\end{proof}

\subsection{Proof of Theorem \ref{theo:no-general-prior}}

\begin{proof}
Since all asymptotically accurate mechanism have a form of the peer truth serum,
it is enough to show that the peer truth serum does not converge if priors are uninformed. 

Consider answer space $X = \{ x, y, z\}$ and suppose at time $t$ a public prior $R^t$ is asymmetric ($R^t[x] \ne R^t[y] \ne R^t[z]$). 
Moreover, suppose agents have fixed prior $Pr$ equal to $Pr[x] = R[x]$, $Pr[y] = R[y] - \epsilon$ and $Pr[z] = R[z] + \epsilon$, where $\epsilon > 0$ is small (smaller than $R[y]$ and $1 - R[z]$),
and posteriors\footnote{If one does not permit beliefs $Pr[\bar x | \bar x] = 1$ and $Pr[\tilde x | \bar x] = 0$, it is enough to redistribute a small value from $Pr[\bar x | \bar x]$ to $Pr[\tilde x | \bar x]$; the proof remains valid. One can apply a similar approach in the proof of Proposition \ref{prop:common-prior}.}
\begin{itemize}
\item $Pr[x|x] = 1$, $Pr[y|x] = 0$, $Pr[z|x] = 0$
\item $Pr[x|y] = 0$, $Pr[y|y] = (Pr[y] + \delta) \cdot k $, 
$Pr[z|y] = (Pr[z] - \delta)\cdot k$, where $k = \frac{1}{Pr[y]+Pr[z]}$ and $0 < \delta << \epsilon$
\item $Pr[x|z] = 0$, $Pr[y|z] = 0$, $Pr[z|z] = 1$
\end{itemize}
In this case, agents can be honest for their observations $x$ and $z$, while
for observation $y$ their reward is:
\begin{itemize}
\item for reporting $y$: $\frac{Pr[y|y]}{R[y]} = \frac{Pr[y] + \delta}{Pr[y] + \epsilon}k < k $, where $\delta$ is much smaller than $\epsilon$.
\item for reporting $z$: $\frac{Pr[z|y]}{R[z]} = \frac{Pr[z] - \delta}{Pr[z] - \epsilon}k  > k$, where $\delta$ much smaller than $\epsilon$.
\end{itemize}
This means that an agent's best response to truthfulness of the other agents is not to be honest when she observes $y$.

Clearly, frequency of $y$ reports is not equal to $R^t[y]$ when agents use symmetric and consistent pure strategies, 
so by the law of large numbers histogram $R$ diverges from its current value $R^t$. On the other hand, if prior $Pr$ is allowed to be uninformed, the true distribution $Q$ can be equal to $Q = R^t$, which means that the payment scheme might not be asymptotically accurate. 
\end{proof}

\subsection{Proof of Proposition \ref{prop:common-prior}}

\begin{proof}
Since all asymptotically accurate mechanisms have a form of the peer truth serum,
it is enough to show that the peer truth serum does not converge if priors can be different. 

Consider a scenario defined by the answer space $X = \{ x, y, z\}$ and the true 
distribution $Q[x] = 0.5$, $Q[y] = 0.2$, $Q[z] = 0.3$. 
Suppose each agent has prior different than $R$, but 
believes that the other agents have priors equal to $R$, and, hence, are truthful.

Let agents have prior and posterior beliefs:
\begin{itemize}
\item When $R[y] < Q[y]$:
\begin{itemize}
\item $Pr[x] = R[x] - \epsilon $, $ Pr[y] = R[y] + \epsilon$, $Pr[z] = R[z]$, where $\epsilon > 0$ is small (smaller than $R[x]$ and $1-R[y]$)
\item $Pr[x|x] = 1$, $Pr[y|x] = 0$, $Pr[z|x] = 0$
\item $Pr[x|y] = 0$, $Pr[y|y] = 1$, $Pr[z|y] = 0$
\item $Pr[x|z] = 0$, $Pr[y|z] = Pr[y] \cdot k  - \delta \cdot Pr[z]$, 
$Pr[z|z] = Pr[z] \cdot k + \delta \cdot Pr[z])$, where 
$k = \frac{1}{Pr[y]+Pr[z]}$ and $0 < \delta << \epsilon$
\end{itemize}
In this case agents are honest for their observations $x$ and $y$, while
for observation $z$ their reward is:
\begin{itemize}
\item for reporting $y$: $\frac{Pr[y|z]}{R[y]} = \frac{Pr[y] \cdot k - \delta \cdot Pr[z]}{Pr[y] - \epsilon} > k + \delta$, where $\delta$ is much smaller than $\epsilon$ ($\delta > 0$ for which the inequality holds always exists, because the strict inequality holds for $\delta = 0$).
\item for reporting $z$: $\frac{Pr[z|z]}{R[z]} = \frac{Pr[z] \cdot k + \delta \cdot Pr[z]}{Pr[z]} = k + \delta$, where $\delta$ is much smaller than $\epsilon$.
\end{itemize}
This means that agents are incentivized to report $y$ instead of $z$.
\item When $R[y] > Q[y]$:
\begin{itemize}
\item $Pr[x] = R[x]$, $Pr[y] = R[y] - \epsilon$, $Pr[z] = R[z] + \epsilon$, where $\epsilon > 0$ is small (smaller than $R[y]$ and $1-R[z]$)
\item $Pr[x|x] = 1$, $Pr[y|x] = 0$, $Pr[z|x] = 0$
\item $Pr[x|y] = Pr[x] \cdot k - \delta \cdot Pr[x]$, $Pr[y|y] = Pr[y] \cdot k + \delta \cdot Pr[x]$, $ Pr[z|y] = 0$, where $k = \frac{1}{Pr[x]+Pr[y]}$ and $0 < \delta << \epsilon$
\item $Pr[x|z] = 0$, $Pr[y|z] = 0$, $Pr[z|z] = 1$
\end{itemize}
In this case agents are honest for their observations $x$ and $z$, while
for observation $y$ their reward is:
\begin{itemize}
\item for reporting $x$: $\frac{Pr[x|y]}{R[x]} = \frac{Pr[x] \cdot k - \delta \cdot Pr[x]}{Pr[x]} = k - \delta$, where $\delta$ is much smaller than $\epsilon$.
\item for reporting $y$: $\frac{Pr[y|y]}{R[y]} = \frac{Pr[y] \cdot k + \delta \cdot Pr[x]}{Pr[y] + \epsilon} < k - \delta$, where $\delta$ is much smaller than $\epsilon$ ($\delta > 0$ for which the inequality holds always exists, because the strict inequality holds for $\delta = 0$).
\end{itemize}
This means that agents are incentivized to report $x$ instead of $y$.
\end{itemize}
  
Under the peer truth serum rewarding policy, whenever $R[x]$ is significantly larger than $Q[x]$ and $R[y] \approx Pr[y]$, the agents act in the following way:
\begin{itemize}
\item report $x$ when they observe $x$ or when they observe $y$ and $R[y] > Q[y]$;
\item report $y$ when they observe $y$ and $R[y] \le Q[y]$ or when they observe $z$ and $R[y] < Q[y]$;
\item report $z$ when they observe $z$ and $R[y] > Q[y]$.
\end{itemize}

Now, if, for example, $R^t[x] = 0.7$, $R^t[y] = 0.2$, and $R^t[z] = 0.1$, $R$ will not converge 
(in probability) towards true distribution $Q$ as the frequency of the reports equal to $z$ is strictly less than $0.3$, while frequency of $x$ is strictly grater than $0.5$.    

\end{proof}

\end{document}